\documentclass[a4paper]{IEEEtran}
\usepackage{graphicx}
\usepackage{amsmath, amsthm, amssymb, amsfonts, array}
\usepackage{colortbl}
\usepackage[normalem]{ulem}
\usepackage{wasysym}
 \usepackage{ulem} 
 
\usepackage[para,online,flushleft]{threeparttable}

\newtheorem{thm}{Theorem}
\newtheorem{asmt}{Assumption}
\newtheorem{prop}{Proposition}

\usepackage[usenames,dvipsnames]{xcolor}
\usepackage{marginnote}

\begin{document}
\title{Demystifying Competition and Cooperation Dynamics of the Aerial mmWave Access Market}

\author{Olga Galinina, Leonardo Militano, Sergey Andreev, Alexander Pyattaev, Kerstin Johnsson, \\
Antonino Orsino, Giuseppe Araniti, Antonio Iera, Mischa Dohler, and Yevgeni Koucheryavy
\thanks{O.~Galinina, S.~Andreev, and Y.~Koucheryavy are with the Department of Electronics and Communications Engineering, Tampere University of Technology, Finland.}
\thanks{L.~Militano, A.~Orsino, G.~Araniti, and A.~Iera are with the ARTS Laboratory, DIIES Department, University Mediterranea of Reggio Calabria,~Italy.}
\thanks{A.~Pyattaev is with YL-Verkot Oy, Finland.}
\thanks{K.~Johnsson is with Intel Corporation, Santa Clara, CA, USA.}
\thanks{M.~Dohler is with the Department of Informatics, King's College London, UK; and Worldsensing, UK and Spain.}
}
\maketitle

\begin{abstract}
Cellular has always relied on static deployments for providing wireless access. However, even the emerging fifth-generation (5G) networks may face difficulty in supporting the increased traffic demand with rigid, fixed infrastructure without substantial over-provisioning. This is particularly true for spontaneous large-scale events that require service providers to augment capacity of their networks quickly. Today, the use of \textit{aerial} devices equipped with high-rate radio access capabilities has the potential to offer the much needed "on-demand" capacity boost. Conversely, it also threatens to rattle the long-standing business strategies of wireless operators, especially as the "gold rush" for cheaper millimeter wave (mmWave) spectrum lowers the market entry barriers. However, the intricate structure of this new market presently remains a mystery. This paper sheds light on competition and cooperation behavior of dissimilar aerial mmWave access suppliers, concurrently employing licensed and license-exempt frequency bands, by modeling it as a vertically differentiated market where customers have varying preferences in price and quality. To understand viable service provider strategies, we begin with constructing the Nash equilibrium for the initial market competition by employing the Bertrand and Cournot games. We then conduct a unique assessment of short-term market dynamics, where two licensed-band service providers may cooperate to improve their competition positions against the unlicensed-band counterpart intruding the market. Our unprecedented analysis studies the effects of various market interactions, price-driven demand evolution, and dynamic profit balance in this novel type of ecosystem.
\end{abstract}

\IEEEpeerreviewmaketitle




\section{Introduction and background}

\subsection{5G standardization update}

The global research on fifth generation (5G) radio access systems has essentially been completed and the respective standardization has actively begun. In light of the recently published IMT vision for 2020 and beyond~\cite{IMT15}, there is now a common consensus that 5G will address scenarios that have difficulty to be served with existing technology, including massive, ultra-reliable, and low latency machine type communications~\cite{Pal16}, as well as enhanced mobile broadband (eMBB) use cases~\cite{Gal15}. For over 4 billion of mobile subscribers in the world, the latter promises to deliver the peak data rates of 10 Gbit/s, which is 100x growth over the corresponding figures in fourth generation (4G) networks.

However, this radical improvement over 4G technology is difficult to materialize without moving up in frequency to harness nontraditional millimeter wave (mmWave) spectrum. Indeed, compared to today's ultra-high frequencies, which comprise only about 1 percent of all regulated spectrum~\cite{Rap14}, mmWave band is plentiful and remains lightly licensed~\cite{Yos16}. Therefore, it could be made available to service operators worldwide, further facilitated by the fact that mmWave frequencies should be 10-100x cheaper per Hz than the conventional spectrum below 3 GHz~\cite{And15}. Hence, mmWave systems may bring along over an order of magnitude increase in capacity even for the existing 4G cell densities~\cite{Ran14}.

Driven by IEEE 802.11ad and emerging 802.11ay specifications, unlicensed mmWave bands at 60 GHz become increasingly employed by commercial Wi-Fi products~\cite{Gal16}. For licensed cellular use, the attractive candidates are 28, 39, and 72 GHz~\cite{Yos16}, and 3GPP accelerate their efforts to ratify a new, non-backward compatible, radio access technology (RAT) operating in mmWave frequencies. The development of this 5G New RAT comprises Phase I (by Sept 2018) optimized primarily for the eMBB scenarios, which then opens door to Phase II (by Dec 2019) supporting all 5G use cases~\cite{IMT15}. Therefore, already in 2018-2019 we expect to see the early commercial deployments of the "initial" New RAT operating in frequencies under 40 GHz. 

\subsection{Dynamic 5G access infrastructure}

The rapid developments in 5G standards provide service operators with efficient means to deploy ultra-dense networks in mmWave bands~\cite{Bal15}, which cope better with the unprecedented acceleration in global mobile traffic demand that is about doubling every year. However, throughout the 40 years of its history, the cellular industry has primarily relied on static radio access network (RAN) deployments. This has led to rigid, fixed network architectures, where it takes a long time to deploy additional access points due to lack of available sites and cumbersome installation procedures. Hence, service providers are actively seeking for alternative technologies and system design options that enable non-rigid placement of access nodes to better accommodate the varying space-time user demand~\cite{Yal16}.

It has recently been understood that autonomous flying robots, nicknamed \textit{drones}, have the potential to quickly deploy dedicated communication networks~\cite{Flo15}, thus bringing access supply to where the demand actually is. Facilitated by miniaturization and cost reduction of electronic components, unmanned aerial vehicles (UAVs) and low-altitude platforms (LAPs) equipped with wireless transceivers may soon lay the foundation for truly dynamic RAN solutions~\cite{Agu16}. Subject to suitable regulatory frameworks, they can not only augment wireless capacity and coverage by meshing with conventional RAN systems and dynamically "patching" them where needed, but also share infrastructure with cargo delivery drones thus bridging across transportation and communication markets.

Owing to their agility and mobility, rapidly deployable drone small cells (or simply \textit{drone cells}) will be particularly useful in 5G networks during unexpected and temporary events, such as large-scale mass outdoor happenings that create unpredictable access demand fluctuations. The examples of such scenarios, as envisioned by~\cite{METIS15}, include an open air festival or a marathon use case with very high densities of users and their connected handheld, wearable devices that altogether produce huge amounts of aggregated traffic. In such areas of interest where conventional RAN infrastructure may be sparse and under-dimensioned, aerial access systems can urgently assist the existing cellular networks by offering the much needed capacity boost.

Fundamentally, the use of drone cells for "on-demand" RAN densification might in the future lead to very different dimensioning of our networks where they are no longer planned for peak loads, but instead provisioned for median loading. Many important challenges have already been resolved to advance this thinking~\cite{Mozaffari2015}, \cite{Xia16}, and in fact Google has recently announced their strategy to offer aerial wireless access in emerging markets\footnote{Alan Weissberger, "Google's Internet Access for Emerging Markets"}. Inspired by these promising scientific, technological, and business advances, \textit{in this work} we put forward the vision of aerial mmWave access networks employing mmWave RAT in high demand and overloaded situations. We argue that mmWave technology is particularly attractive for aerial access due to a number of factors, such as simpler air interface design exploiting large bandwidths~\cite{Gho14} and higher chances of having a line-of-sight link to the user on the ground, which is crucial for efficient mmWave operation.

As follows from the above, aerial mmWave access systems may soon deliver very high spectral efficiencies and ultra-broadband service experience, hence attracting various operators to leverage the cheaper mmWave spectrum. However, the structure of this \textit{unprecedented market} is not nearly well understood and remains in prompt need of a comprehensive research perspective. Fortunately, the broad field of game theory supplies us with a rich set of tools to address its dynamics and study the intricate interactions between the existing players that may want to adapt their respective business strategies as well as the new market stakeholders. We continue with a review of our selected game theoretic mechanisms in what follows.

\subsection{Game theory for aerial mmWave market}

The ongoing evolution of wireless networks fueled by the recent advent of novel access technologies and complex heterogeneous architectures calls for applying advanced game theoretic tools to help design and model the appropriate 5G-ready solutions~\cite{han2011game}. We are convinced that the emerging 5G mmWave access market structure and interactions remain a vastly unexplored area as of today. Accordingly, various market players concurrently employing licensed and license-exempt mmWave spectrum will have to take into account the evolved expectations of their customers when offering new services. Therefore, we target a careful modeling of viable business strategies in the aerial mmWave access market by systematically bringing into focus the aspects of the quality of service (QoS) and pricing, mindful of customer preferences.

\vspace{-15px}
\begin{figure} [!ht]
  \begin{center}
    \includegraphics[width=0.51\textwidth]{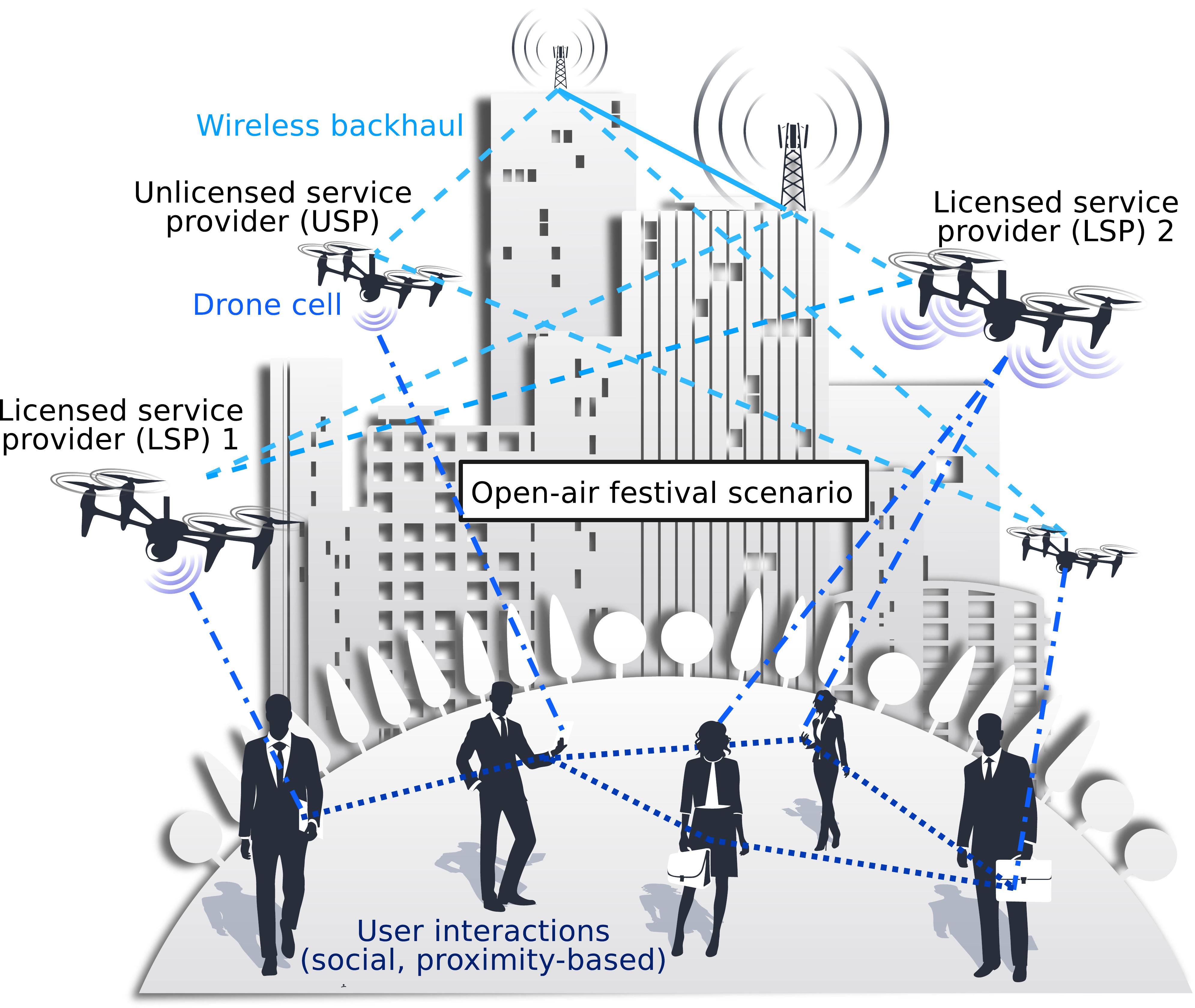}
  \end{center}
  \caption{Our motivating scenario with 5G drone cells.}
  \label{fig:scenario}
\end{figure} 

Our first line of research considers a \textit{duopoly} setting where two licensed-band service providers (LSPs) enter the aerial mmWave access market and compete with each other to maximize their profits (see Fig.~\ref{fig:scenario}). In the past, market competition and pricing models have already attracted significant attention across the wireless community, as reported in~\cite{ren2013entry} and \cite{niyato2008competitive}. To the best of our knowledge, however, no prior contribution in existing literature on game theory has been considering \textit{vertical differentiation} in the context of wireless networking where market players differentiate their products in terms of quality, as was studied in~\cite{shaked1982relaxing}, \cite{gabszewicz1979price}, and \cite{shaked1983natural}. 

This methodology is particularly interesting for the purposes of our analysis since~\cite{shaked1982relaxing} showed that for some parameters the game has a unique subgame perfect equilibrium at which only two players enter the market. Furthermore, these two suppliers choose to offer differentiated products and earn positive profits at equilibrium (if the production costs are disregarded). Different from the work in~\cite{shaked1982relaxing}, we assume that the LSPs have already decided to enter the market and are instead interested in defining the optimal quality of their products as well as the corresponding prices/quantities~\cite{motta1993endogenous}. 

Therefore, our methodology at its initial stage studies market equilibria in which the LSPs first determine the specification of their offered products (i.e., mobile subscriptions) and then decide on the prices or the quantities of the products they sell according to a \textit{Bertrand} or a \textit{Cournot} competition model, respectively~\cite{bonanno1986vertical}. Another finding of this paper is in understanding how the outcome of the initial LSP competition is influenced by the customer preferences (so-called \textit{taste parameter}) willing to pay more or less for a higher quality product. After this is done, we introduce another entity into our system, an unlicensed-band service provider (USP) operating in free-to-use mmWave spectrum, and then study the \textit{dynamics} of the resulting market.
 
At the dynamic stage of our considered game, the customers are allowed to prefer the USP to the LSP for wireless access services, while previously inactive users can activate to connect to the USP. The decision strategy of the customers is tightly coupled with their individual \textit{utility} perception, and the resulting user decisions determine the \textit{strategy revision protocol}~\cite{sandholm2009pairwise}. The latter is captured in our methodology as a combination of different subjective and social factors, such as "curiosity", "dissatisfaction", and "gossiping". In practice, the competing market players attempt to find countermeasures against customer decisions that might penalize their profits. While for the USP we model its ability to adapt the offered price over time according to the user choices, the LSPs do not have such luxury as they are bound by a long-term contractual agreement with their customers. 

Correspondingly, the LSPs have to seek ways to maximize the chances of meeting their service-level agreements (SLAs) and thus prevent users from changing the provider \textit{in the long run}. To this aim, we advocate the adoption of long-term cooperative agreements between the LSPs  to improve the QoS perception levels for their customers. Such cooperation is particularly helpful if a customer could be served by an assisting LSP at much higher spectral efficiency. Our analysis shows that these forms of cooperation between the LSPs improve their profits as compared to the non-cooperative case. Then, the question is how to share the surplus among the cooperating players, which requires \textit{fairness}-centric game theoretic solutions~\cite{tijs1986game}. In this work, we adopt the widely-accepted Shapley value~\cite{shapley1952value} for its intrinsic capability to capture the contributions of individual players in a coalition~\cite{militano2016enhancing}.

Finally, we note that market evolution and price dynamics were investigated in other contexts as well, including heterogeneous small cell networks~\cite{rose2014pricing} and advanced offloading techniques~\cite{zhu2014pricing}. Evolutionary games~\cite{hofbauer2003evolutionary}, \cite{Wei97} have also been receiving attention in the literature, with applications to competing wireless operators~\cite{korcak2012competition} and user groups sharing the limited access bandwidth~\cite{niyato2009dynamics}. In contrast to the past work, our game formulations uniquely study the evolution of market shares for two LSPs and one USP -- the minimal feasible set of players in the subject market -- according to the customer decisions as well as the ability of the USP to adapt its pricing strategy over time. Technically, we construct a system of differential equations that characterize the dynamics of our market by analyzing the situations when the LSPs avoid or enter into cooperation while competing with the USP.

\subsection{Contributions of this work}

We envision that the aerial mmWave access technologies will play a pivotal role in the emerging 5G market. To leverage their full potential, they will need to be accompanied by new powerful methodologies able to capture the intricate market dynamics and characterize the practical benefits for all the involved stakeholders. In this paper, we offer the first comprehensive attempt to analyze and understand the competition and cooperation behavior of dissimilar players within this novel ecosystem, which boils down to the following \textit{four major contributions}. 

\begin{enumerate}
\item \textbf{Initial market formulation.} We introduce viable strategies of the LSPs (licensed-band service providers) when operating in the new market of the aerial mmWave access systems. A vertically differentiated market is comprehensively modeled where two LSPs first determine the specification of their offered services, and then decide on the prices or the quantities of the services they offer according to the Bertrand or Cournot competition models. Our proposed formulation enables a valuable comparison of the equilibrium points established with the two considered initial games.

\item \textbf{Analysis of market dynamics.} We construct an unprecedented analytical framework that captures the short-term market dynamics featuring dissimilar players, two LSPs and one USP (unlicensed-band service provider) as the latter enters the subject market. Our dynamic game theoretic methodology allows to carefully follow the cooperative interactions between the LSPs and their competition against the USP counterpart, as well as thoroughly characterize both price- and sociality-driven evolution of customer preferences when making the service provider selection.

\item \textbf{Numerical performance evaluation.} We systematically report important numerical results for the considered market players under the realistic assumptions on the mmWave channel modeling, signal propagation, and aerial access point (AAP) operation. Our findings open door to an exciting technology innovation when two AAPs belonging to different LSPs serve customers in cooperation, thus allowing to share the access infrastructure between multiple operators. A detailed study of temporal evolution in market shares and resulting profits is delivered for both the Bertrand and Cournot competition models.

\item \textbf{Large-scale system-level validation.} We verify our extensive analytical results with in-depth system-level simulations that are grounded in reality and constructed after the meaningful real-world scenarios. The employed simulation platform integrates considerable knowledge behind the principles of mmWave system operation as well as employs substantiated and adequate assumptions on the AAP deployment behavior. The simulation results are made available to validate the core assumptions of our proposed analytical framework, assess the true market dynamics, and support the key practical learnings.
\end{enumerate}



\section{Proposed system model} \label{system_model}
In this section, we outline the considered system as well as introduce its main assumptions. 

\subsection{Scenario of interest}

We focus on an open air, densely crowded scenario (e.g., a festival or any other outdoor event) within a certain area of interest, where participants are located on the ground. The mobile devices of these people -- \textit{potentially active} in terms of wireless access -- are equipped with radio transceivers and are able to operate on either \textit{licensed} or \textit{unlicensed} mmWave frequencies. 
\begin{asmt}
For the purposes of mmWave channel modeling, we represent a human body as a cylinder of height $h_b$ and diameter $2r_b$, and further assume that positions of the centers of cylinders are spatially distributed with the density of $\mu$. We note that not all of our mass event participants may use mmWave radios and also assume a certain share of potentially active ones, termed \emph{customers}. Further, the average elevation of the mobile devices equals $h_d$, and the projections of the \textit{active devices} on the ground are also spatially distributed with the respective density of $\mu_0 << \mu$. 
\end{asmt}

Mobile devices of customers on the ground may be served by aerial mmWave access points (or drone cells) that belong to a certain owner (i.e., service provider). These flying access nodes may be deployed temporarily within the area of interest to support very high but relatively short-term connectivity demand on the ground during the event in question. 
\begin{asmt}
We assume that $N_i$ AAPs (aerial access points) within a certain service provider's deployment $i$ (termed here a fleet) are identical, \textit{uniformly placed} over the area, keep the same altitude $h_i$, and operate over the spectrum bandwidth $B_i$. An AAP may serve a customer's device with a beam of half-angle $\phi$ (half of aperture), which is directed at the inclination angle $\beta \leq \beta_{\max}$, where $\beta_{\max} < \pi/2$ follows from physical restrictions (i.e., the angle between the vertical and the beam cone axes cannot exceed a certain maximum).
\end{asmt}

If a certain device is associated with the fleet $i$ based on a particular customer agreement (the agreements are all indivisible and mutually exclusive), it will be served by the closest AAP of the corresponding fleet. Let us consider a link of the tagged device associated with the fleet $i$. We assume that one of the following alternative situations may take place (see e.g.,~\cite{Mozaffari2015}):
\begin{itemize}
\item there is line-of-sight (LOS) propagation between the device and the corresponding AAP,
\item the device is blocked by the body of its owner or another person, but there still exists a sufficiently strong path for the reflected beam,
\item the device is fully blocked and thus has no usable connection to the AAP.
\end{itemize}

In order to address the above options, we introduce the following assumptions.
\begin{asmt}\label{asm:powers}
Signal propagates according to a free-space model and hence, depending on the link between the AAP and the device, the received power at distance $d \leq h \tan \beta_{\max}$ may assume one of the following values:
\begin{equation}
\left\{ \begin{array}{l}
p_{rx} =  p_{tx}G_{a}\frac{G_i}{h^2+d^2},\text{ if LOS exists},\\
p_{rx} = p_{tx} G_{NLOS} G_{a} \frac{G_i}{h^2+d^2},\text{ if no LOS, but signal is reflected},\\
p_{rx} = 0, \text{ if no signal at all},
\end{array} \right.
\end{equation}
where $p_{tx}$ is the fixed transmit power, $G_{a} = \frac{2}{1-\cos \phi}$ is the antenna gain, $G_i = \left( \frac{c}{4\pi f_i}\right)^2$ is a path gain constant, $f$ is the signal frequency, $c=3 \cdot 10^8$ is the speed of light. Here, $G_{NLOS}$ is a constant decrease in signal power due to reflection (additional attenuation factor as in~\cite{Mozaffari2015}), which is assumed to be fixed across the network.

\end{asmt}
Based on the received power, we may estimate the actual throughput of the device connected to the AAP of the fleet $i$ and located at the distance $d$ according to the Shannon's formula:
\begin{equation}
T = \Delta B_i \log_2\left(1+ \frac{p_{rx}}{N_0}\right),
\end{equation}
where $N_0$ is the noise plus interference level and $\Delta B_i$ is a share of effective bandwidth available to the link in question. We assume that each device always uses the full channel bandwidth (which may vary depending on the choice of a provider), while the time is shared between all the connected devices equally (e.g., as a result of the round-robin scheduling as in~\cite{zhu2014pricing}). Owing to the beamforming capability on the mmWave AAPs, we can analyze our scenario as noise-limited, and thus $N_0$ is assumed constant.

\begin{asmt}\label{asm4}
In addition, we assume that the number of AAPs within each fleet is sufficient to cover the area of interest. In particular, the distance between the two AAPs is $2R_{\text{AAP}}$, $ R_{\text{AAP}} \leq h \tan \beta_{\max}$. This readily implies a lower bound on the number of AAPs (for a hexagonal grid): 
\begin{equation}
\begin{array}{l}
\!N_{\min} \!=\! \left \lceil \frac{R}{2 h \tan \beta_{\max}}\!+\!1\!-\!\sqrt{2} \right \rceil  \left \lceil \frac{R\sqrt{\frac{4}{3}}}{2 h \tan \beta_{\max}\!+\!1\!-\!\sqrt{\frac{8}{3}}} \right \rceil \!+\!\\
\! +\!\left \lfloor \left \lceil \frac{R\sqrt{\frac{4}{3}}}{2 h \tan \beta_{\max}+1-\sqrt{\frac{8}{3}}} \right \rceil \right \rfloor. \!\!\!\!
\end{array} \label{eqn:drone_number}
\end{equation}
\end{asmt}
The above expression (or similar, depending on selected deployment) is merely an example and may be obtained based on straightforward geometric reasoning; thus, we omit its derivation here to save space.

\subsection{Main players and their interactions}
In this work, we differentiate between the following \textit{three types} of players in our system: (i) multiple customers demanding service, (ii) two LSPs (licensed-band service providers) operating in \textit{licensed} mmWave spectrum and "supposedly" responsible for the announced QoS, and (iii) one USP (unlicensed-band service provider) operating in free-to-use \textit{unlicensed} mmWave frequencies and aiming to offer best-effort service. The interactions of interest between these players also belong to the following \textit{three types}:
\begin{description}
\item[ A.] A tagged customer may be served by its own LSP paying a \textit{fixed} price $p_i,i=1,2$ (per unit time). For that matter, the LSP provides the customer with a mobile subscription (e.g., a SIM-card) in advance, where the choice of an LSP is made based on the individual customer preferences (see below).
\item[ B.] Another alternative for the customer is to be served by the USP paying a \textit{dynamic} price $p_0(t)$. We note that all of our devices with mmWave radios may be served by the USP, including those that do not have an LSP subscription provided in advance.
\item[ C.] Finally, an LSP may request \textit{assistance} from another LSP and utilize its infrastructure to serve a certain customer e.g., if the latter is geometrically closer to other provider's AAP (see Fig.~\ref{fig:beamforming} for details). In this case, the served customer is unaware of such cooperation and continues to pay its regular price $p_i,i=1,2$. 
The LSPs may share the resulting surplus according to a certain partitioning model e.g., using the Shapley vector~\cite{shapley1952value}. 
\end{description}

\begin{figure} [!ht]
  \begin{center}
    \includegraphics[width=0.5\textwidth]{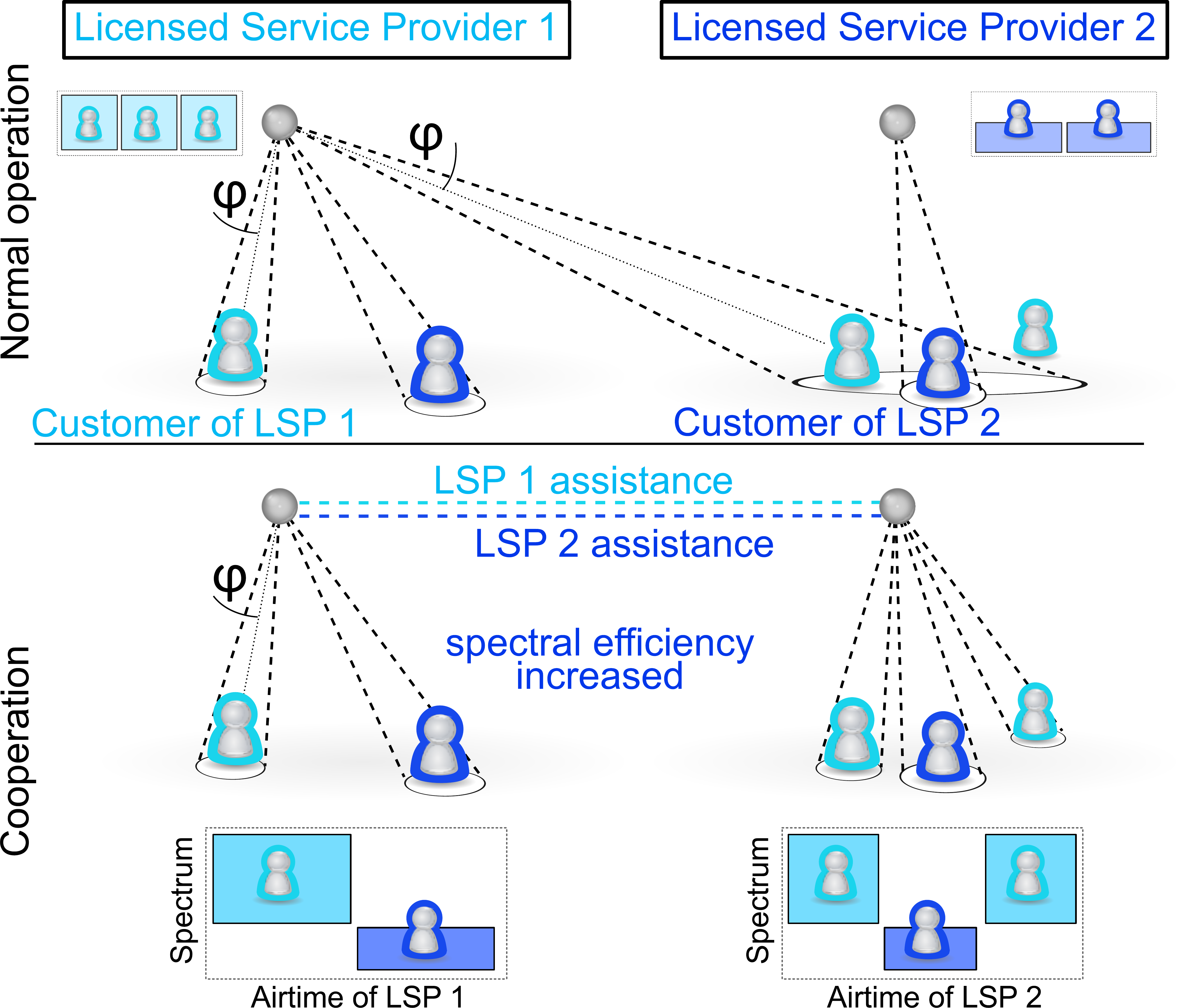}
  \end{center}
  \caption{Cooperation between service providers: mutual assistance of the LSPs.}
  \label{fig:beamforming}
\end{figure} 

\subsection{Strategies and payoffs}
Below we discuss the strategies and the corresponding payoffs of our three types of players.

\begin{enumerate}
\item \textbf{Both LSPs} aim to support the desired QoS and keep their subscribers satisfied. 
To do so, one LSP may offload some of its customers (connections) to another LSP, and thus have them served on its own spectrum, but using the airtime of the assisting LSP. In such a situation, the assisting LSP's infrastructure will effectively act as a "proxy" AAP.

\textbf{Strategy of the LSP $i$} includes offering prices $p_{i}$ per unit time for its customers as well as announcing a certain QoS (e.g., throughput) at the initial stage of the market game.

\item \textbf{USP's} profit comprises the fees paid by all of its customers (potentially including former customers of either LSP and those without the LSP subscriptions).

\textbf{Strategy of the USP} includes offering price $p_0(t)$ for providing best-effort service to its customers. 

\item Each of $N$ \textit{active} customers aims at maximizing its utility $U(T,p)$, that is, a function of price and experienced throughput. 

\textbf{Strategy of a customer} includes choosing a service provider (e.g., its own LSP, if applicable, or the USP). Accordingly, a customer may prefer its own LSP (in case a SIM-card has been purchased in advance) or connect to the USP instead. 

\end{enumerate}

\subsection{Customer's preferences}
In contrast to most past work, we assume in this paper that customers are not identical in their preferences, which conveniently reflects their different financial capabilities, tastes, or priorities. We thus consider a vertically differentiated market~\cite{Lancaster1990}, where potential customers agree on ranking diverse products (in our case, mmWave access offers) in the order of quality preference according to some utility function. However, their willingness to pay remains variable due to e.g., budget restrictions. 

\begin{asmt} 
In our scenario, we introduce the following utility function:
\begin{equation}
U(\theta, T, p) =  \theta \cdot s(T) - p, 
\label{eqn:utility} 
\end{equation}
where $s$ is an increasing quality function determined by the throughput $T$ that the customer may observe while making a decision (i.e., either announced by the provider at the initial stage or actually experienced during operation), $p$ is the price per airtime, and $\theta$ is the taste parameter~\cite{motta1993endogenous}. The latter is a typical measure of the buyer's preference in the differentiated markets. The higher its taste parameter is, the more a customer is willing to pay for a better quality service. Generally, zero utility value corresponds to not buying the product, and thus if $U(\theta, T, p)\leq0$, then a customer would decide to refrain from purchasing the mmWave access service. We assume that $\theta$ is distributed within the interval $[0,\theta_{\max}]$ according to a certain probability density $h_{\theta}(\theta)$. 
\end{asmt}

Our customers are \textit{rational}, that is, they always make decisions targeting the better utility value. Here, zero utility corresponds to the situation when a customer decides not to connect, while zero lower bound on $\theta$ implies that there always exists someone who decides not to connect at all (i.e., we assume that our setup is not "covered market" in order to be able to apply the Cournot competition model further on). 

\subsection{Customer quality estimation function} 
We continue by discussing the shape of the function $s(T)$, which reflects the dependence of customer service satisfaction levels on the throughput: announced (expected) or realistic (experienced). Building on our supportive analysis of today's mobile data offers (3G/4G) across the leading network operators in the US and Europe (see Table~\ref{tab:1}), we conclude that the price per unit throughput follows \textit{exactly} the same law across a wide range of larger and smaller providers alike, which implies:
\begin{equation}
\begin{array}{c}
\frac{p}{T} = \frac{a}{T+b} + c.
\end{array}
\end{equation}

The above follows from the appropriate curve fitting and the corresponding regression parameters as well as the values of R-coefficient are shown in Table~\ref{tab:1}. 
\begin{asmt}\label{asm:7}
Since the user utility is assumed to be linear with respect to the price as well as the quality $s(T)$, we impose that $s(T)$ has a similar structure to what service providers use in real markets. Accordingly, our quality $s(T)$ translates into:
\begin{equation}
\begin{array}{c}
s(T)= \frac{aT}{T+b} + cT,
\end{array}
\end{equation}
where $a$, $b$, and $c$ are the coefficients that e.g., may be taken from Table~\ref{tab:1}. 
\end{asmt}

Without loss of generality, we note that the coefficient $a$ could be included into $\theta$, and the coefficient $c$ is then rescaled correspondingly. This effectively means that the maximum price, which a customer is ready to pay for a certain quality $T$ is $p_{\max}^{\theta} = \theta a \cdot  (\frac{T}{T+ b} + \tilde {c}T)$, and $a$ is then incorporated into $\tilde \theta = \theta a$. Therefore, the upper border of taste $\theta_{\max}$ yields a limit on price, where price $\tilde \theta_{max} \cdot  (\frac{T}{T+b} + cT)$ is the lowest among those leading to zero demand. In what follows, we omit tilde over $\theta$ for improved readability. 

\begin{table}[tbp]
\caption{Price per unit throughput: comparing selected providers.}
	\centering
	\begin{tabular}{|c|c|c|c|c|}
\hline
Operator  & a &b &c&R-coefficient\\
\hline
AT\&T (US)\footnotemark[1]  & 16.35 &-0.0273 &6.7120&0.9997\\
Verizon (US)\footnotemark[2] & 41.99 &0.5793 &3.3980&0.9997\\
T-Mobile (Germany)\footnotemark[3] &12.50 &2.175e-08 &3.7500&1.0000\\
Vodafone (UK)\footnotemark[4] & 12.95 &0.1134 &2.3690 &1.0000\\
Orange (France)\footnotemark[5] & 10.80 &1.6000 &2.0000&1.0000\\
Telecom Italia (Italy)\footnotemark[6] & 19.69 &2.5000 &0.6250&1.0000\\
DNA (Finland)\footnotemark[7] &12.36 &0.5619 &0.1216&1.0000\\
\hline
	\end{tabular}
  \label{tab:1}
  \begin{tablenotes}
\raggedright
\item[1] https://www.att.com/shop/wireless/plans/planconfigurator.html\\
\item[2] http://www.verizonwireless.com/landingpages/cell-phone-plans/\\
\item[3] http://www.t-mobile.com/cell-phone-plans/mobile-internet.html\\
\item[4] https://www.vodafone.co.uk/shop/bundles-and-sims/sim-only-deals/\\
\item[5] http://boutique.orange.fr/tablette-et-cle/forfaits-internet-let-s-go\\
\item[6] https://www.tim.it/offerte/mobile/internet-da-tablet-e-pc\\
\end{tablenotes}
\end{table}



\section{Initial stage of our game formulation}

 \subsection{General remarks}
We decompose our subsequent modeling into two consecutive parts, which are the (i) \textit{initial-stage} game formulation where two LSPs divide the market in advance and (ii) \textit{dynamic} game development after the USP enters the market in equilibrium.

\subsubsection{Initial stage} At the initial, \textit{long-term} stage, only two major players (the LSP~1 and the LSP~2) are participating in a differentiated market game with \textit{two phases}: first, deciding on the maximum quality level to attract customers and second, competing in price or, alternatively, in quantity. We consider both price and quantity competitions as they lead to different equilibrium points, and there is still no consensus in current literature as to which type of competition is preferred. At the point of equilibrium, a certain share $D_1 \in (0,1)$ of customers have purchased subscriptions (SIM-cards) of the LSP~1, while $D_2 \in (0,1-D_1)$ have acquired SIM-cards of the LSP~2. The remaining customers may be not connected to either LSP, since our setting is not a "covered market".  We briefly summarize the above as:
\begin{itemize}
\item \textbf{Players}: two LSPs.
\item \textbf{Strategies}: first -- qualities, then -- prices/quantities.
\item \textbf{Payoff}: utility of each LSP (subscriber payments minus costs).
\end{itemize}

\subsubsection{Dynamic stage} At the second, \textit{short-term} stage, a new market player (the USP) is intruding the system where the LSPs operate at the equilibrium point. The USP announces its dynamic price $p_0(t)$ at time moment $t$. Based on their utility function, customers of the LSP $i$ may decide to prefer the USP, while previously non-participating customers (having share $y_0 \in [0,1-D_1-D_2]$) may connect \textit{only} to the USP. Our system tracks the shares $x_1$, $x_2$, $y_0$, $y_1$, and $y_2$, where $x_i$ correspond to customers currently utilizing service of the LSP $i$, and $y_i$ denotes those who change their LSP $i$ for the USP. 

We assume that every customer of the LSP $i$ is aware of both prices $p_0(t)$ and $p_i$ as well as knows its current throughput ($T_0(t)$ or $T_i(t)$). 
Further, customers are able to interact with each other (owing to their close proximity within the area of interest) and hence may \textit{imitate} strategies of others (see e.g., \cite{korcak2012competition}). In our scenario, a customer of the LSP $i$ or USP with the taste $\theta$ decides to inquire another customer from the potential group of the USP or LSP $j$ about the experienced service quality. 
Based on the received information $T_j$, the tagged customer calculates its own potential utility $U(\theta;s(T_j),p_j)$ and, if the potential service is better than the current one, changes its service provider with the probability that is proportional to the difference $U(\theta;s(T_j),p_j)-U(\theta;s(T_i),p_i)$.

Summarizing, the considered dynamic game may be described by:
\begin{itemize}
\item \textbf{Players}: USP and the customers.
\item \textbf{Strategies}: $p_0(t)$ for the USP and service provider selection for the customers.
\item \textbf{Population state}: shares of the customers $(y_0,y_1,y_2)$.
\item \textbf{Payoff}: utility of the USP (pure profit) and utility of the customers.
\end{itemize}
The overall structure of the dynamic game is illustrated in Fig.~\ref{fig:game}.

\begin{figure} [!ht]
  \begin{center}
    \includegraphics[width=0.5\textwidth]{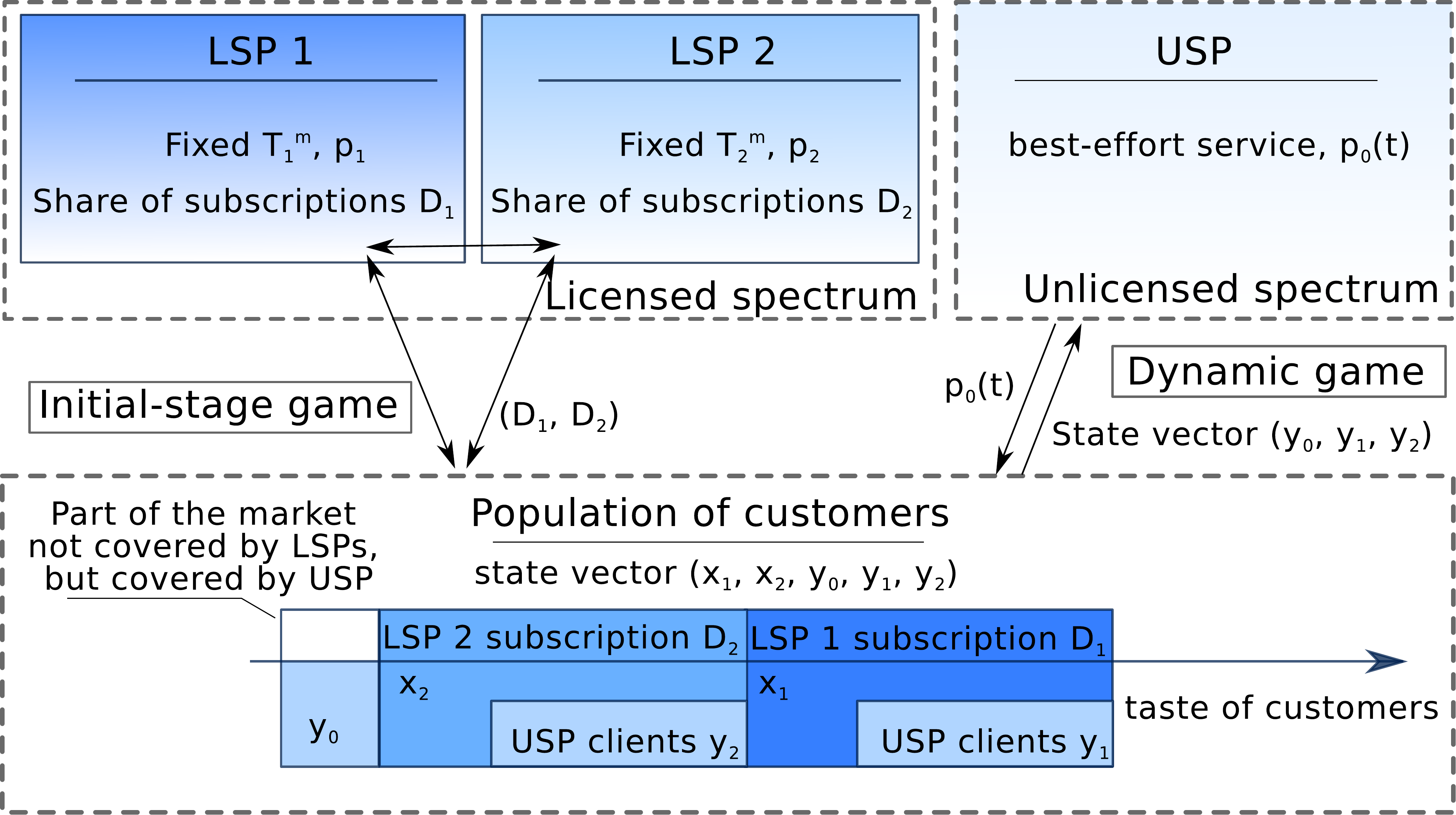}
  \end{center}
  \caption{Interactions in our considered dynamic game.}
  \label{fig:game}
\end{figure} 


\subsection{Profits of the LSPs at the initial stage}




First, two LSPs play a two-stage game to divide the market and set the optimal quantity or price (here, we do not consider the market entry stage~\cite{Gal1983} for the LSPs). For the sake of further comparison, we model \textit{both alternatives}: price setting and quantity setting types of competition, which are known in the literature as Bertrand and Cournot competition models, correspondingly. In this work, we analyze both game types in order to reveal the dependence of the results on the optimal choice of LSP strategies, namely, whether LSPs eventually offer a homogeneous product (as shown by the Cournot model) or two differentiated products (as illustrated by the Bertrand model). Since both situations may happen in the real market, one model cannot be preferred over another.

At the \textbf{first} stage, both LSPs select the offered level of quality. We recall that quality is characterized by a data rate package with \emph{provisionally announced maximum throughput} $T_i^{m}, i=1,2$, and the LSPs are aware of each other, but make their decisions sequentially. At the \textbf{second}, competitive stage, in the Bertrand game the LSPs decide on the prices $p_i, i=1,2$ that will be announced to the customers purchasing their SIM-cards (a mutually exclusive choice is assumed here, that is, either a SIM-card of the LSP~1 or that of the LSP~2). In contrast to that, in the Cournot game the LSPs decide on quantity, which in our scenario translates into the number of subscribed customers or, equivalently, sold SIM-cards.

Making its decisions at both stages, an LSP aims to maximize its \textbf{profit}. Clearly, the LSP income is determined by the financial flow from the previously subscribed and presently served customers, but to define the rest of the profit function we need to clarify the structure of the costs. Clearly, the advertised service quality may not be matching the real quality, which uninformed customers face after purchasing the product. However, in our scenario, maintaining subscription (loyalty) to a certain LSP in a long run is equivalent to a "repeat purchase". Therefore, sellers are highly motivated in keeping their subscribed customers satisfied, as inspired by theory of advertising in~\cite{Milgrom1986}, \cite{Kihlstrom1984}. To reflect this motivation analytically, we introduce the following assumption responsible for handling the costs of the LSPs.

\begin{asmt}
We assume \emph{linear costs} of improvement in the claimed quality $s_i$ \emph{per user}, so that an LSP would be ready to support it for its served customers. Hence, the total costs depend both on the number of served devices and on the announced $T_i^{m}$. The latter reflects the initial investments into a fixed-term spectrum lease and/or the amounts of spectrum that could be resold (as e.g., in~\cite{Bennis2008} or~\cite{niyato2008competitive}) or, otherwise, rented. 
\end{asmt}

In fact, the dependence of costs is not linear on the required throughput, because in practice not all of the customers are active simultaneously and the resources of the LSPs are carefully provisioned. Moreover, the quality function $s(T)$ has a non-linear component $\frac{T}{T+b}$, which, however, has less impact than the linear term $cT$ and can be omitted for the sake of tractability of further analysis. Therefore, our proposed profit function is given by:
\begin{equation}
\begin{array}{c}
\Pi_{i} (s_1,s_2,p_1,p_2)= D_i(s_1,s_2, p_1, p_2) \left(p_i -  \nu s_i \right),
\end{array}
\end{equation}
where $D_i(s_1,s_2, p_1, p_2)$ is a demand for the LSP $i$, and $\nu$ is a "quality cost coefficient". The latter may be roughly estimated from the value of the spectrum license costs to support the announced QoS, normalized by unit time as well as the total number of customers in the country of interest.

\subsection{Demands for the LSPs at the initial stage}
To establish the Nash equilibrium of our game at its initial stage, we exploit the powerful principles of \textit{backward induction}. Assuming that the levels of quality $s_i = s(T_i^{m})$ (or, equivalently, the levels of the announced throughput $T_i^{m}$) are fixed, we first obtain the equilibrium for the Bertrand price game and the Cournot quantity game. Without loss of generality, we may assume that $s_1 \geq s_2$ ($T_1^{m} \geq T_2^{m}$), and hence $p_1 \geq p_2$, since the sellers are targeting the rational customers, and we consider \emph{one-dimensional} product differentiation (otherwise, equation (\ref{eqn:utility}) would contain additional components related to other criteria~\cite{Vandenbosch1995}). 

Following~\cite{motta1993endogenous}, and based on the fixed price and quality levels, we determine the following points of indifference for a tagged customer:
\begin{itemize}
\item point of indifference to buying or not buying the service of the LSP~1 is denoted by the parameter $\theta_{\varnothing,2} = \frac{p_2}{s_2}$ (follows from $U(\theta, s_1, p_1) = 0$),
\item point of indifference to buying the service of the LSP~1 or the LSP~2 corresponds to the parameter $\theta_{1,2} = \frac{p_1 - p_2}{s_1-s_2}$ (follows from $U(\theta, s_1, p_1)  = U(\theta, s_2, p_2)$).
\end{itemize}

Therefore, the corresponding demands may be written as:
\begin{equation}
\begin{array}{c}
D_1(s_1,s_2, p_1, p_2) = \int \limits_{\theta_{1,2}}^{\theta_{\max}} h(\theta) d \theta=H(\theta_{\max})-H(\theta_{1,2} ), \\
D_2(s_1,s_2, p_1, p_2) = \int \limits_{\theta_{\varnothing,2}}^{\theta_{1,2} } h(\theta) d \theta =H(\theta_{1,2} )-H(\theta_{\varnothing,2}),
\end{array}
\label{egn:DEMAND}
\end{equation}
where $h(\theta)$ and $H(\theta)$ are the PDF and the CDF of the taste parameter $\theta$, respectively. In what follows, we consider the Bertrand price competition and the Cournot quantity competition models separately.

\subsection{Bertrand price competition}
In the Bertrand game, the LSPs choose prices $p_1$ and $p_2$ in order to maximize their profits $\Pi_{i} = D_i \left(p_i -  \nu s_i \right)$ for the selected quality function values $s_1$ and $s_2$. 

\begin{asmt}
For the purposes of illustration in this paper, we only consider a uniform distribution of the taste parameter $h(\theta) = \frac{1}{\theta_{\max}}, \theta\in [0,\theta_{\max}]$. In case of more than one player in the market, if $h(\theta)$ is not uniform, we may experience multiple symmetric equilibria~\cite{Gal1983}, which could be investigated separately and may thus generalize the results of this paper in subsequent studies.
\end{asmt}

Differentiating $\Pi_{i} $ over $p_i$, one may calculate the optimal prices (can be verified for $\nu=0$ by~\cite{motta1993endogenous}) for the fixed levels of quality:
\begin{equation}
\begin{array}{c}
p_1 = s_1\frac{\nu(2s_1-s_2) +2\theta_{\max}(s_1-s_2)}{(4 s_1 - s_2)} , \quad
p_2 = s_2 \frac{3\nu s_1 +\theta_{\max}(s_1-s_2)}{(4s_1-s_2)}, 
\end{array}
\end{equation}
where $s_1 =s( T_1^{m})$, $s_2 =s(T_2^{m})$. This derivation is trivial and is based on finding points where the gradient of $\Pi_{i}$ is zero. Then, the profit function is given as:
\begin{equation}
\begin{array}{c}
\Pi_{1} (s_1,s_2)=  4s_1^2\frac{(\theta_{\max}-\nu)^2(s_1-s_2)}{\theta_{\max}(4 s_1 - s_2)^2},\\ \\
\Pi_{2}  (s_1,s_2)= s_1 \frac{(\theta_{\max}-\nu)(-4\nu s_1+\theta_{\max}s_2)(s_1-s_2)}{\theta_{\max}(4s_1-s_2)^2}.
\end{array}
\end{equation}

The first-order condition of maximum for the above optimization functions may be further written as follows:
\begin{equation}
\begin{array}{c}
\frac{\partial \Pi_{1}}{\partial s_1} = 4 s_1(\theta_{\max}-\nu)^2\frac{4s_1^2 - 3s_1s_2 + 2s_2^2}{\theta_{\max}(4s_1 - s_2)^3} =0,\\ 
\frac{\partial \Pi_{2}}{\partial s_2} =s_1^2 (\theta_{\max}-\nu)^2\frac{4s_1 - 7s_2}{\theta_{\max}\left(4s_1 - s_2\right)^3}= 0.
\end{array}
\end{equation}

Denoting $\frac{s_1}{s_2}$ as $x$, we may then locate the maximum points for both LSPs. We note that due to the absence of roots for the first equation and the fact that $\frac{\partial \Pi_{1}}{\partial s_1}>0$, the maximum is located at the border $s_1^* = s_{\max}$, while the optimal quality $s_2^* = s_{\max}\xi$, where $\xi =4/7$ (the second-order condition of maximum $\left. \frac{\partial ^2\Pi_{2}}{\partial s_2^2}\right |_{s^*_1,s^*_2}<0$ could be verified easily). Hence, we have established a candidate solution for the Bertrand game (here, $\theta_{\max}>\nu$ should hold).

\begin{prop}
The obtained solution for the Bertrand game is the Nash equilibrium.
\end{prop}
\begin{proof}
The proof is fairly straightforward and is based on demonstrating that the following holds:
\begin{equation}
\begin{array}{c}
\Pi_i(s^*_1,s^*_2) \geq \Pi_i(s_1,s_2^*) ,\quad\text{ for any } s_1 < s^*_1,\\
\text{ and }\Pi_i(s^*_1,s^*_2) \geq \Pi_i(s_1^*,s_2) ,\quad\text{ for any }  s_2 \neq s_2^*.
 \end{array}
\end{equation}
\end{proof}

Substituting the sought point $(s_{\max}, s_{\max}\xi)$ into the price, demand, and profit function, we may also obtain the indicators, which are collected altogether in Table~\ref{tab:2}. Then, we additionally calculate the \textit{consumer surplus}, which may be derived as:
\begin{equation}
\begin{array}{c}
CS= \int \limits_{\theta_{1,2}}^{\theta_{\max}} (\theta s_1-p_1) \frac{1}{\theta_{\max}} d \theta +
 \int \limits_{\theta_{\varnothing,2}}^{\theta_{1,2} }(\theta s_2-p_2) \frac{1}{\theta_{\max}} d \theta =\\
 =\frac{2s_1^2 s_2(\nu - \theta)^2 (3\nu + \theta)}{3(4s_1 - s_2)^2\theta^2(2\nu + \theta)}.
 \end{array}
\end{equation}

\subsection{Cournot quality competition}
While in the Bertrand market the price $p_i$ is controlled by the LSP $i$ and the share of connected users is then determined through the demand function, in the Cournot market the LSPs control the quantity (i.e., the number of users) and then the prices are derived through the inverted system of demand functions:
 \begin{equation}
\begin{array}{c}
p_1(s_1,s_2,D_1,D_2) = -\theta(D_1s_1 - s_1 + D_2s_2),\\
p_2(s_1,s_2,D_1,D_2) = -\theta s_2(D_1 + D_2 - 1).
 \end{array}
\label{eqn:PRICES_C}
\end{equation}

Substituting the above into the expression for the LSP profit, we may establish the quantity response functions that maximize the profit for the fixed qualities $s_1$, $s_2$ based on the following first-order conditions:
\begin{equation}
\begin{array}{c}
\frac{\partial \Pi_{1}}{\partial D_1} =-\theta(2D_1s_1 - s_1 + D_2s_2)-\nu s_1,\\
\frac{\partial \Pi_{2}}{\partial D_2}=-\theta s_2(D_1 + 2D_2 - 1)-\nu s_2.
 \end{array}
\nonumber
\end{equation}

Therefore, the sought demand functions that maximize the LSP profit are given by:
\begin{equation}
\begin{array}{c}
D_1(s_1,s_2)  = \frac{(\theta_{\max}-\nu)(2s_1 - s_2)}{\theta_{\max}(4s_1 - s_2)},\quad
D_2(s_1,s_2) = \frac{(\theta_{\max}-\nu) s_1}{\theta_{\max}(4s_1 - s_2)}.
 \end{array}
\nonumber
\end{equation}

Substituting these demand functions into (\ref{eqn:PRICES_C}), we obtain the prices set by the LSPs:
\begin{equation}
\begin{array}{c}
p_1 = \frac{\theta s_1(2s_1 - s_2)+2\nu s_1^2}{4s_1 - s_2}, \quad 
p_2 = \frac{\theta s_1s_2+\nu(3s_1 - s_2)}{4s_1 - s_2}
 \end{array}
\nonumber
\end{equation}
and, correspondingly, capture the resulting profit:
\begin{equation}
\begin{array}{c}
\!\!\!\Pi_1\!(s_1\!,s_2) \!=\! \frac{\theta s_1(2s_1 \!- \!s_2)^2(\theta_{\max}\!-\!\nu)^2}{\theta_{\max}(4s_1 \!- \!s_2)^2},
\Pi_2 \!(s_1\!,s_2)\!=\!\frac{\theta s_1^2s_2(\theta_{\max}\!-\!\nu)^2}{\theta_{\max}(4s_1 - s_2)^2}.\!\!
 \end{array}
\label{eqn:PROFIT}
\end{equation}

In the second stage of backward induction, we derive the optimal level of qualities that maximize the profit (\ref{eqn:PROFIT}) by finding the stationary points of the following equations:
\begin{equation}
\begin{array}{c}
\frac{\partial \Pi_{1}(s_1,s_2)}{\partial s_1} = (\theta_{\max}-\nu)^2 \frac{4s_1 (4s_1^2 -3s_1s_2^2 + s_2^3)}{4s_1 - s_2)^3},\\
\frac{\partial \Pi_{2}(s_1,s_2)}{\partial s_2} =(\theta_{\max}-\nu)^2 \frac{s_1^2(4s_1 + s_2)}{\theta_{\max}(4s_1 - s_2)^3}.
 \end{array}
\label{eqn:PROFIT_deriv}
\end{equation}

Denoting $s_1/s_2$ as $x$, we may conclude that there exists no solution $x>1$ for (\ref{eqn:PROFIT_deriv}). Since both $\frac{\partial \Pi_{1}(s_1,s_2)}{\partial s_1}$ and $\frac{\partial \Pi_{2}(s_1,s_2)}{\partial s_2} >0$, the point of maximum is located at the right border of the interval for $s$, that is, $s_1^* = s_{\max}$ and $s_2^* = s_{\max}$. The final prices, profits, as well as points of indifference are easy to verify and are summarized in Table~\ref{tab:2}. Therefore, we have established a candidate solution for the Cournot game and can formulate a proposition similar to the one before as follows.

\begin{prop}
The obtained solution for the Cournot game is the Nash equilibrium.
\end{prop}
\begin{proof}
The proof is easy to derive similarly to that of the above proposition for the Bertrand game.
\end{proof}

Since the Cournot prices and qualities are equivalent, two LSPs divide the subject market in equal proportions, if we assume that there is no weighted preference towards a certain brand. Hence, the consumer surplus in this case may be defined as:
\begin{equation}
\begin{array}{c}
CS= \int \limits_{\theta_{1,2}}^{\theta_{\max}} (\theta s_1-p_1) \frac{1}{\theta_{\max}} d \theta=
\frac{7 s_1^2 s_2 (\nu - \theta_{\max})^2(12\nu + 5\theta)}{12(4s_1 - s_2)^2\theta_{\max}^2 (7\nu + 5\theta)}.
 \end{array}
\end{equation}

Finally, the difference between the Bertrand and the Cournot games with respect to the consumer surplus equals:
\begin{equation}
\begin{array}{c}
CS_{\text{Bertrand}} -CS_{\text{Cournot}}= \frac{5s_1^2s_2 (\nu - \theta)^2 (10 \nu + \theta)}{24\theta(4s_1 - s_2)^2(14 \nu^2 + 9\nu\theta + \theta^2)}>0.
  \end{array}
\end{equation}

As a result of all the above derivations, we obtain the \textit{divided} (albeit not covered) market at time moment $t=0$, which is illustrated in Fig.~\ref{fig:line}. In the figure, the upper part corresponds to our static initial solution. 

\begin{figure} [!ht]
  \begin{center}
    \includegraphics[width=0.5\textwidth]{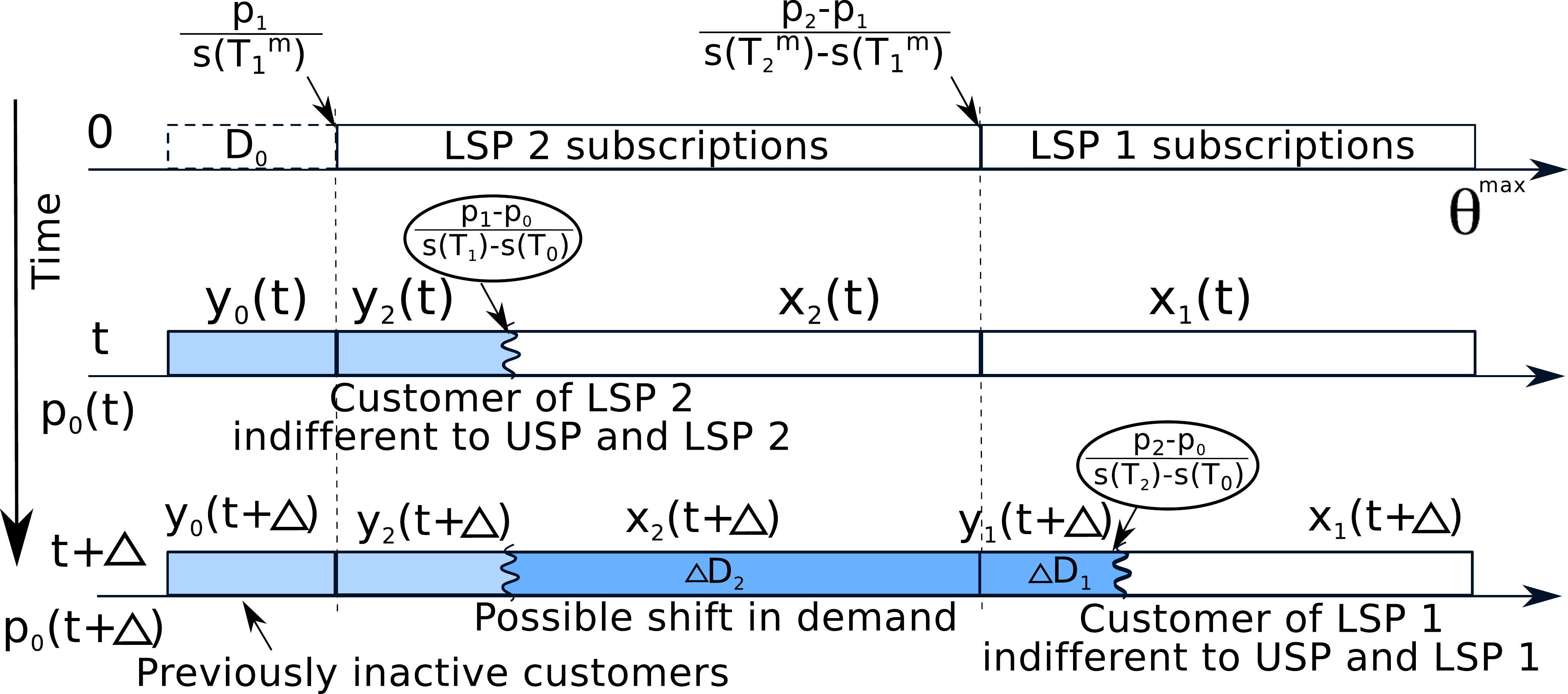}
  \end{center}
  \caption{Illustration of our  considered vertically differentiated market.}
  \label{fig:line}
\end{figure} 



\begin{table}[tbp]
\caption{Comparing our initial game solutions for various scenarios.}
	\centering
	\begin{tabular}{|c|c|c|}
\hline
    	
		Game type& Bertrand & Cournot   \\
		\hline
		Quality of LSP~1\footnote{2}    &$s$											           & $s$   \\
		Quality of LSP~2                      &$0.5714s$											  & $ s$\\
\hline
		Point $\theta_{\varnothing,2}$&$0.8750 \nu +0.125\theta$								  & $0.6667  \nu+0.3333\theta$   \\
		Point $\theta_{1,2}$                &$ 0.5833 \nu +0.4167\theta$							  & $ 0.6667  \nu+0.3333\theta$   \\
		\hline
		Profit of LSP~1                         &$0.1458 s \frac{(  \theta-\nu )^2)}{ \theta}   $				  & $0.1111 s \frac{(  \theta-\nu )^2}{ \theta} $   \\
		Profit of LSP~2                         & $0.0208 s \frac{( \theta-\nu) (  \theta-7.0 \nu )}{ \theta}  $            &   $0.1111 s \frac{(  \theta-\nu )^2}{ \theta} $       \\
		Aggregate MP\footnote{3}      &$0.1667 s \frac{(7 \nu^2 + 4  \theta^2 - 11 \nu  \theta)}{4 \theta}$  & $0.2222 s \frac{(  \theta-\nu )^2}{ \theta}  $       \\  
		\hline
		Price of LSP~1                         &$0.2500 s (3 \nu +  \theta)  $                   					  & $0.3333 s  (2\nu + \theta)$        \\
		
		Price of LSP~2                         &$0.0714 s (7 \nu +  \theta)  $  							  & $0.3333 s  (2\nu + \theta)$     \\
		\hline
		Costs of LSP~1		             &$0.5833\nu s \frac{( \theta-\nu )}{\theta} $                   					  & $0.3333\nu s\frac{( \theta-\nu )}{\theta} $        \\
		Costs of LSP~2		             &$0.1667\nu s\frac{( \theta-\nu )}{\theta} $                   					  & $0.3333\nu s\frac{( \theta-\nu )}{\theta} $        \\
		\hline
		Demand of LSP~1		             &$0.5833\frac{( \theta-\nu )}{\theta} $                   					  & $ 0.3333\frac{( \theta-\nu )}{\theta} $        \\
		Demand of LSP~2		             &$0.2917\frac{( \theta-\nu )}{\theta}  $                   					  & $0.3333\frac{( \theta-\nu )}{\theta} $        \\
		\hline
		Total demand                         &$0.8750 \frac{( \theta-\nu )}{\theta} $						  &$0.6667\frac{( \theta-\nu )}{\theta}$        \\
		\hline
		CS                        &$\frac{2s_1^2 s_2(\nu - \theta)^2 (3\nu + \theta)}{3(4s_1 - s_2)^2\theta^2(2\nu + \theta)} $						  &$\frac{7 s_1^2 s_2 (\nu - \theta)^2(12\nu + 5\theta)}{12(4s_1 - s_2)^2\theta^2 (7\nu + 5\theta)}$        \\
		
		\hline
	\end{tabular}
  \label{tab:2}
  \begin{tablenotes}
\raggedright
\item[1]Cooperative strategy in Cournot game (both LSPs maximize total profit).\\
\item[2]For brevity, $\theta = \theta_{\max}$, $s = s_{\max}$.\\
\item[3]Market profit.
\end{tablenotes}
\end{table}

\section{Dynamic development of game formulation}
In this subsection, we introduce a new player on the market, namely, the USP. In our scenario, the USP enters the market where the two LSPs are operating in equilibrium, for a \textit{short-term} period, thus distorting the balance. Assessing the changes in the customer alignment picture constitutes the theoretical target of this section. We note that in the considered market the customers are differentiated by their experienced spectral efficiency (due to the actual geometry of their locations) as well as their preferences with respect to the price/throughput trade-off (according to our assumption of vertically differentiated market). To tackle the resulting complex and multidimensional problem, we capture the system dynamics by averaging the spectral efficiencies across the coverage areas of the AAPs, which may be equivalent to characterizing customer mobility.

\subsection{Dynamic setup and mmWave properties}
In order to follow the dynamic evolution of customer utility, we first focus on deriving the average spectral efficiency, which is a function of the number of simultaneously served customers per mmWave AAP. For that matter, we consider a \textit{tagged} AAP of the fleet $i$ representing the average system behavior within this fleet. Owing to the flexibility of drone cell deployment as well as minding our earlier assumption of uniform AAP placement,  we further assume that the coverage area of this AAP may be approximated by a circle of radius $R_{AAP}$ that is equal to a half of the distance between two neighboring AAPs (as per Assumption \ref{asm4}):
\begin{equation}
\begin{array}{l}
R_{AAP} = \frac{R}{\sqrt{2}+  \left \lceil   \frac{R} {2 h_i\tan \beta_{\max} }+1-\sqrt{2}  \right \rceil   -1}.
\end{array} 
\end{equation}

We note that the above has been derived by analogy to (\ref{eqn:drone_number}) reflecting a particular considered deployment, and is meant here to solely serve the purposes of radius estimation illustration.
\begin{asmt}
Specifically, we assume that user devices are distributed according to a Poisson Point Process (PPP) with the density $\mu_0$, which has been introduced earlier in Assumption 1.
\end{asmt}
Then, the number of devices per AAP may be calculated as:
\begin{equation}
\begin{array}{c}
\!\!\!\! n_i \!=\! \mu_0 \pi R_{\text{AAP}}^2  x_i,i\!=\!1,2, \quad n_{0} \!=\! \mu_0 \pi R_{\text{AAP}}^2 (y_0\!+\!y_1\!+\!y_2),\!\!
\end{array} 
\end{equation}
where $x_i,i=1,2$ correspond to the shares of customers for the LSPs and $y_i,i=0,1,2$ is the total share of customers for the USP.

\begin{asmt}
For the sake of analytical tractability, we replace random spectral efficiencies of the customers with a spectral efficiency value averaged across the coverage area of the tagged AAP.
\end{asmt}

Due to our PPP assumption, the distribution of distances $d$ to the AAP within a circle of radius $R_{\text{AAP}}$ equals $f_d(x)=\frac{2x}{R_{\text{AAP}}^2}$. Therefore, the average spectral efficiency $\overline{\eta}(R_{\text{AAP}})$ may be established by calculating the following integral:
\begin{equation}
\begin{array}{c}
\overline {\eta} (R_{\text{AAP}}) = \int \limits_{0}^{R_{\text{AAP}}}  \log_2\left(1+ \tilde p(q_{\text{LOS}}+q_{\text{NLOS}}G_{NLOS})\right) \frac{2x}{R_{\text{AAP}}^2} dx,
\end{array} 
\end{equation}
where $\tilde p=p_{tx}\frac{G_i G_a}{(h^2+d^2)N_0}$ for brevity, $G_i$ is the path gain defined in Assumption~\ref{asm:powers}, $G_a$ is the antenna gain, and the probabilities of LOS/NLOS link have different nature by contrast to the conventional approach in~\cite{ITU1410}. These are derived in Appendix and summarized here as: 
\begin{equation}
\left\{ \begin{array}{l}
q_{\text{LOS}}(d) = e^{-\mu d \cdot 2r_b \left( \frac{h_b-h_d}{h-h_d}\right)},\text { if LOS exists},\\
q_{\text{NLOS}|\text{no LOS}}(d)=\left(1-\exp({- \mu {\frac{\pi \left (h \tan{(\phi)} \right)^2 \cdot \cos (\phi)}{\sqrt{\cos^2(\phi)-\frac{d^2}{d^2+{h^2}}} }}} )\right)  \times \\
\times  \left (1-q_{\text{LOS}}(d) \right), \text { if no LOS, NLOS exists}.
\end{array} \right.
\end{equation}

The \textit{closed-form} approximation is based on the fact that $E[f(x)] \approx f(E[x])$ and is given by:
\begin{equation}
\begin{array}{c}
\overline {\eta} (R_{\text{AAP}}) =   \log_2\left(1+ p_{tx}\frac{G_i G_a(q_{\text{LOS}}+q_{\text{NLOS}}G_{NLOS})}{\left(h^2+\left(\frac{2}{3R_{\text{AAP}}}\right)^2 \right)N_0} \right) ,
\end{array} 
\end{equation}
where the above numerical solution remains suitable for the necessary further calculations. To assist in the process, Fig.~\ref{fig:SE_plot} illustrates the changes in the slope of spectral efficiency as averaged over a circle of varied radius for $28$ and $60$~GHz. It clearly indicates the impact of blockage at lower AAP altitudes as well as shows the maximum distance to the receiver.

\begin{figure} [!ht]
  \begin{center}
    \includegraphics[width=0.50\textwidth]{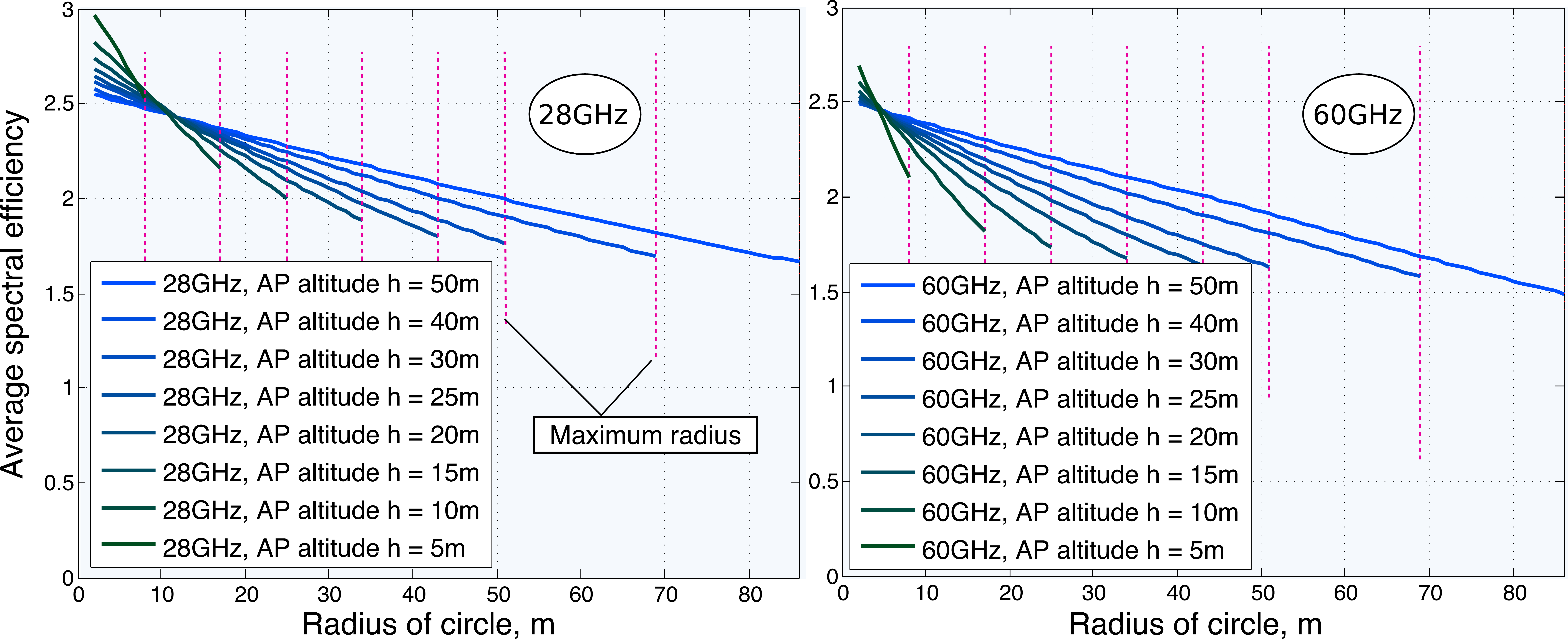}
  \end{center}
  \caption{Spectral efficiency averaged over a circle of varied radius.} 
  \label{fig:SE_plot}
\end{figure} 

\subsection{USP market entry} 
In our formulation, the population state may be described by the vector $(y_1,y_2,y_0)$, which denotes the shares of customers connected to the USP, having purchased a SIM-card of the LSP~1 and the LSP~2, correspondingly, or having no SIM-card at all. Further, the shares of the customers actively using the service by the LSPs may be derived directly from the initial-stage demand $x_1(t)= D^*_1 - y_1(t)$, $x_2= D^*_2 - y_1(t)$ (we remind that the customers are not allowed to change their SIM-cards during our dynamic short-term game). 

Given the information on $(y_0(t),y_1(t),y_2(t))$, the USP attempts to make its optimal decision on the dynamic price $p_0(t)$ that maximizes its total profit (in the absence of full market information), which in turn directly depends on both the number of connected customers and the price itself. For the fixed price $p_0$ and the estimated average actual throughput $\frac{B_{SP} \eta_{SP}}{y_1+y_2+y_0}$ per customer, we may establish the USP profit as: 
\begin{equation}
\Pi_{0} = p_0(t) D_{0}(t) = p_0(t) \left( y_0(t)+ y_1(t) +  y_2(t)\right), 
\label{eqn:USP_profit}
\end{equation}
where $y_0$ denotes the proportion of the customers coming from a previously inactive share of the market (those who did not purchase any SIM-card before).


Further, let the USP be unaware of future customer behavior and our hidden customer reaction functions, so that the USP could not maximize (\ref{eqn:USP_profit}) explicitly. Then, its \textit{heuristic choice} would be to react flexibly to the changes in the customer demand. In particular, if the number of subscribers is decreasing, then the USP lowers the price and vice versa. Therefore, the stationary point for the price $p(t)$ is defined by the constant demand of the USP:
\begin{equation}
\begin{array}{c}
\frac{d(y_0(t)+y_1(t)+y_2(t))}{dt}=0.
\end{array}
\end{equation}

Further, we assume that the considered price dynamics set by the USP satisfies the following equation:
\begin{equation}
\begin{array}{c}
\frac{dp}{pdt}= c_{\text{price}}\frac{d(y_0(t)+y_1(t)+y_2(t))}{dt},
\end{array}
\end{equation}
that is, the relative increment of the USP price is proportional to the increment in the number of currently served customers. We note that the latter is adopted as one of the reasonable examples where the USP acts rationally, and other USP strategies may in principle be considered.

\subsection{Decisions of the customers} 
In order to mimic customer decisions when there is a choice between two utility values, we introduce the following assumption:
\begin{asmt}
We assume that if there is a utility difference of $\Delta U(\theta)$ between its possible value and the current one, the customer decides to change or imitate the strategy with the probability:
\begin{equation}
\begin{array}{l}
p_u\left (\Delta U(\theta) \right)  = \frac{2}{1+e^{-c_u\Delta U(\theta)}}-1, \text{ for } \Delta U(\theta) > 0.
\end{array}
\end{equation}
The above consideration tightly approximates the behavior of customers with arbitrary levels of patience, as given by the "impatience" coefficient $c_u$. In particular, $p_u\left (0 \right) = 0$ (no changes whatsoever) and $p_u\left ( \Delta U \rightarrow \infty \right) = 1$. This logistic function is a convenient formulation to reflect the decision making probability in the absence of the information on the actual range of utility values within our game.
\end{asmt}

Modeling the customer behavior in more detail, we employ a notion of proximity-based social interaction. Accordingly, the customers do not know the actual quality of services they had never used before, but they may "talk" to other customers with a certain degree of $\gamma$ (could be regarded as a proximity-based interaction possibility). It effectively means that at some point of time the tagged customer may select another customer from the alternative service group to inquire for its service quality. When the USP announces the price $p_0(t)$ at the state $(y_0,y_1,y_2)$, there could be three types of \textit{mutually exclusive} reactions of the LSP customers as follows.

\begin{enumerate}
\item \textbf{Gossiping.} At the moment $t+\Delta t$, a customer of the LSP $i$ decides to assess its quality 
and thus inquires someone connected to the USP, if any, with the probability $\gamma$. Knowing the alternative quality $s(T_0)$ at time $t$, said customer may decide to change its current service provider, if $U(\theta;s(T_0), p_0) > U_i(\theta;s(T_i), p_i)$ with the probability $p_u(\Delta U(\theta))$. The corresponding probability is defined as:
\begin{equation}
\begin{array}{l}
\xi \gamma (y_0+y_1+y_2) p_u  \!  \left (U(\theta;s_0, p_0) - U(\theta;s_i, p_i)\right),
\end{array}
\end{equation}
where $s_i = s(T_i)$ for brevity.
\item \textbf{Curiosity.} Another alternative is to connect to the USP out of "human curiosity", which is here independent of either social opinion or own quality observation, but dictated by human nature to explore new opportunities. We note that our earlier assumption on rationality of the customers still holds here, since curiosity aims solely at improving the utility. For the "curiosity" factor $\alpha$, the corresponding probability is given by:
\begin{equation}
\begin{array}{c}
\xi  (1-\gamma)  \alpha_c.
\end{array}
\label{eqn:curiosity}
\end{equation}
\item \textbf{Dissatisfaction.} The last considered type of user reactions is when the customer is highly dissatisfied with the current service -- which e.g., happened to be far worse than what was expected -- and decides to connect to the USP with the probability comprising the "dissatisfaction" factor $\delta$ that corresponds to the actual level of disappointment, which leads us to:
\begin{equation}
\begin{array}{c}
\!\! \!\!\!\!\!\!\! \xi  (1\!-\!\gamma)(1\!-\!\alpha_c) \delta p_u \! \left( U(\theta;s(T_i^{m}), p_i)\!-\!U(\theta;s_i, p_i)  \right)\!.
\end{array}
\end{equation}
\end{enumerate}

We note that the probability $\xi$ to consider changing the current service provider is but a technical parameter here, which regulates the rate of the process in question as well as determines how frequently the USP price is reconsidered. The parameters $\alpha$, $\gamma$, and $\delta$ are assumed fixed for all of the customers. Furthermore, for backward transitions from the USP to the LSP service, the above three options hold as well:

\begin{enumerate}
\item \textbf{Gossiping.} For a customer already connected to the USP, the "gossiping rate" towards its initial subscription is estimated as:
\begin{equation}
\begin{array}{l}
\xi \gamma \cdot x_i p_u  \!  \left (U(\theta;s_i, p_i) - U(\theta;s_0, p_0)\right).
\end{array}
\end{equation}
\item \textbf{Curiosity.} The "curiosity rate" remains the same as in (\ref{eqn:curiosity}).
\item \textbf{Dissatisfaction.} The corresponding probability based on the level of dissatisfaction is given by:
\begin{equation}
\begin{array}{c}
\!\!\!\! \!\!\!\!\! \!\! \xi (1\!-\!\gamma)(1\!-\!\alpha_c) \delta p_u \! \left( U(\theta;s(T_i^{m}), p_i)\!-\!U(\theta;s_0, p_0)  \right)\!,
\end{array}
\end{equation}
where $T_i^{m}$ is the announced throughput and $s(T_i^{m})$ is the solution to the initial-stage game of the LSPs.
\end{enumerate}

We note that previously inactive customers do not have expectations and thus may select the USP at the moment $t$. Hence, they do so if, after they inquire for the relevant information, the new utility value is above zero, that is, with the probability $ \xi \gamma  p_u \! \left( U(\theta;s_0, p_0)  \right)$, or out of curiosity $\xi  (1-\gamma)  \alpha$, since they do not expect any QoS guarantees. Also, these users disconnect from service with the probability $\xi p_u\left( -U(\theta;s_0, p_0)  \right)$, if at a certain time moment their utility drops below zero.

\subsection{Capturing system dynamics} 
In order to capture the dynamic evolution of both the \textbf{customer} and the \textbf{USP} strategies, we further model the behavior of the customers on the market independently of their personal preferences $\theta$, by averaging across the corresponding market shares. Even though within one group the customers still tend to act non-uniformly, this approximation appears to be surprisingly accurate (as we confirm with simulations later on) to model the overall market evolution. First, we derive all the needed expressions for describing said dynamics and then separately calculate the corresponding coefficients (i.e., rates $Q^{\text{X}}_{i \rightarrow j}$, see below).

We continue by constructing a system of differential equations, according to which our market is evolving with time:
\begin{equation}
\left\{ \begin{array}{l}
\frac{dy_1}{dt} =      (D_1-y_1) \xi  \gamma (y_0+y_1+y_2) \cdot  Q_{1\rightarrow 0}^{\text{G}}            \quad \text{   // gossiping}  \\
+ (D_1-y_1)\xi \alpha_c  (1-\gamma)                       									\quad \text{   // curiosity} \\
+   (D_1-y_1)\xi \delta (1-\gamma) \cdot Q^{\text{D}}_{1\rightarrow0}  							 \quad  \text{   // dissatisfaction}\\
 -   y_1  \xi        \gamma (D_1-y_1)\cdot  Q_{0\rightarrow1}^{\text{G}} 						 \quad \text{   // gossiping} \\
 - y_1\xi            \alpha_c  (1-\gamma)  											\quad \text{   // curiosity}\\
 - y_1 \xi            \delta (1-\gamma)\cdot Q^{\text{D}}_{0\rightarrow1},  							\quad \text{   // dissatisfaction}  \\
 \\
\frac{dy_2}{dt} =   (D_2-y_2)\xi   \gamma (y_0+y_1+y_2) Q_{20}^{\text{G}} 				\quad \text{   // gossiping} \\
 + (D_2-y_2) \xi              \alpha_c  (1-\gamma) 										\quad \text{   // curiosity} \\
 +  (D_2-y_2) \xi   \delta (1-\gamma) 	\cdot Q^{\text{D}}_{2\rightarrow0} 						 \quad  \text{   // dissatisfaction} \\
  - y_2  \xi            \gamma (D_2-y_2) \cdot Q_{02}^{\text{G}} 								 \quad \text{   // gossiping}\\
 -y_2 \xi          \alpha_c     (1-\gamma)											 \quad  \text{   // curiosity}\\
 - y_2 \xi             \delta (1-\gamma) \cdot Q^{\text{D}}_{0\rightarrow2} , 						\quad \text{   // dissatisfaction}  \\
\\
\frac{dy_0}{dt} =   (D_0-y_0)\xi\gamma (y_0+y_1+y_2) \cdot  Q_{\varnothing \rightarrow 0} 			 \quad \text{   // gossiping} \\
+(D_0-y_0) \xi \alpha_c  (1-\gamma) 											 \quad \text{   // curiosity} \\
-  y_0   \xi \cdot Q_{0\rightarrow \varnothing } ,\text{   // low utility} 
\\
\\
\frac{dp_0}{pdt}= c_{\text{price}}\frac{d(y_0+y_1+y_2)}{dt},\\
\end{array} \right.
\label{eqn:main_dynamics}
\end{equation}
where the price update coefficient $c_{\text{price}}$ of the USP corresponds to a unit of time, $Q_{i\rightarrow 0}^{\text{G}}$ (or $Q_{0\rightarrow i}^{\text{G}}$) are the coefficients reflecting the group-average "willingness" to change the service provider from the LSP $i$ to the USP (or backwards) after inquiring proximate users, while $Q^{\text{D}}_{i\rightarrow0}$ ($ Q^{\text{D}}_{0\rightarrow i}$) is the average "dissatisfaction" rate due to a difference between the quality announced by the LSP and the actual experienced quality. Further, $Q_{\varnothing \rightarrow 0} $ ($Q_{0\rightarrow \varnothing } $) corresponds to a decision on connecting (disconnecting) for the customers without the SIM-cards. Also here, $0\leq y_1 \leq D_1$, $0\leq y_2 \leq D_2$, $0\leq y_0 \leq D_0$, and $D_i$ are the initial shares for the LSPs obtained with either Bertrand or Cournot solutions, while $D_0 = 1-D_1 - D_2$ are the initially inactive customers. 


\subsubsection{Changing service provider after proximate interaction}
The current throughput of the customer from the group $i$ is based on equal sharing of bandwidth, that is, $T_i = \frac{B_{i} \overline{\eta}_{i}}{N x_i},i=1,2 $ or $T_0 = \frac{B_{0} \overline{\eta}_{0}}{N (y_0+y_1+y_2)} $, 
where $\overline{\eta}_{i}, i=0,1,2$ is the average spectral efficiency over the coverage area of one mmWave AAP. We may then calculate $s_i(T_i) = \frac{1}{\frac{b}{T_i}+1}+cT_i$ for each $i=0,1,2$ and, denoting these functions as $s_i$, 
continue by deriving the coefficients $Q_{i\rightarrow 0}^{\text{G}}$ and $Q_{0\rightarrow i}^{\text{G}}$ that capture system dynamics under user "gossiping". We also note that:
\begin{equation}
 \begin{array}{c}
U(\theta;s_0, p_0) - U(\theta;s_i, p_i)) = \theta (s_0-s_i)-(p_0-p_i).
\end{array}
\end{equation}

Hence, denoting $\frac{p_0-p_i}{s_0-s_i}$ as $\theta_{0,i}$, which defines the ranges of $\theta$, we obtain:
\begin{equation}
 \begin{array}{c}
\! \! \! \! \! \! U(\theta;s_0, p_0) - U(\theta;s_i, p_i)) >0, \quad\quad\quad\quad\\
\quad \quad \quad \quad  \text{ if }  \theta > \theta_{0,i} ,s_0>s_i,\text{ or } \theta < \theta_{0,i},s_0<s_i,
\end{array}
\end{equation}
and a similar expression can be established for $U(\theta;s_0, p_0) - U(\theta;s_i, p_i)) <0$. Further, in order to average over the range of $\theta$, we integrate $p_u\left (\Delta U(\theta) \right)$ separately for $\Delta U(\theta)>0$ and $\Delta U(\theta) <0$ as:
\begin{equation}
\begin{array}{l}
\!\!\Delta U\!(\theta)\!>\!0:  \frac{1}{\theta_{\max}} \!\! \int\! \left( \frac{2}{1+e^{-c_u (\theta (s_0-s_i)-(p_0-p_i))}}-1 \right) \!d \theta \!= \!\\
\frac{\theta}{\theta_{\max}} \!+ \!\frac{1}{\theta_{\max}}  \frac{2}{s_0-s_1}\log(e^{-c_1 \theta+c_2} + 1),\\
\!\!\Delta U\!(\theta)\!<\!0:  \frac{1}{\theta_{\max}} \!\! \int \!\left( \frac{2}{1+e^{-c_u (\theta (s_0-s_i)-(p_0-p_i))}}-1 \right) \!d \theta \!=\! \\
\frac{\theta}{\theta_{\max}} \!-\! \frac{1}{\theta_{\max}}  \frac{2}{s_0-s_1}\log(e^{c_1 \theta-c_2} + 1).
\label{eqn:integrals}
\end{array}
\end{equation}

Based on the above range, we may now average over the possible values of $\theta$, and write the following:
\begin{equation}
\begin{array}{l}
Q_{i\rightarrow 0}^{\text{G}} = \frac{1}{\theta_{\max}}  \int \limits_{A_1}^{A_2} \left( \frac{2}{1+e^{-c_u (\theta (s_0-s_i)-(p_0-p_i))}}-1 \right) d \theta = \\
\left[\frac{A_2-A_1}{\theta_{\max}} + \frac{1}{\theta_{\max}}  \frac{2}{s_0-s_1}\log\frac{e^{-c_1 A_2+c_2} + 1}{e^{-c_1 A_1+c_2} + 1}\right],
\end{array}
\end{equation}
where $A_1\geq A_2$ constitute the integration range depending on $i$ as well as on the type of the initial game as per Table~\ref{tab:paramsA}. If $A \geq A_1$, we say that $Q_{i\rightarrow 0}^{\text{G}} = 0$.

\begin{table}[tbh]
\centering 
\caption{Considered system evaluation parameters}
\begin{tabular}{|c|c|c|c|c|}
\hline \label{tab:paramsA}
  & \multicolumn{2}{c|}{$s_0>s_i$}  & \multicolumn{2}{c|}{$s_0<s_i$\footnotemark[1]}\tabularnewline
  \hline   
  & Bertrand &    Cournot   & Bertrand   &  Cournot   \tabularnewline
\hline  \hline
$A_1$, LSP 1  &  \multicolumn{2}{c|}{$\theta_{\max} $ }                                         & \multicolumn{2}{c|}{$ \min(\theta_{\max},\theta_{0,i})$}    \tabularnewline \hline 
$A_2$, LSP 1  &  \multicolumn{2}{c|}{$\max (\theta_{1,2} ,\theta_{0,i} )$   }            & \multicolumn{2}{c|}{ $ \theta_{1,2} $  }                              \tabularnewline
\hline 
\hline 
$A_1$, LSP 2  & $\theta_{1,2} $                                                         &$ \theta_{\max}$         &$\min (\theta_{1,2},\theta_{0,i} )$                    &$\min (\theta_{\max},\theta_{0,i} )$\tabularnewline \hline 
$A_2$, LSP 2  & \multicolumn{2}{c|}{$\max (\theta_{2,\varnothing} ,\theta_{0,i} )$ }           & \multicolumn{2}{c|}{$\theta_{2,\varnothing}$}\tabularnewline
\hline 
\end{tabular}
\begin{tablenotes}
\raggedright
\item[1] If $s_0=s_i$, a difference between the utilities is only defined by $p_0-p_i$.
\end{tablenotes}
\end{table}

Similarly, $Q_{0 \rightarrow i}^{\text{G}} $ may be obtained with the only difference that the columns $s_0>s_i$ and $s_0<s_i$ swap, and the second line in (\ref{eqn:integrals}) is utilized for $A_1 \geq A_2$ (otherwise, $Q_{0 \rightarrow i}^{\text{G}} =0$):
\begin{equation}
\begin{array}{l}
Q_{0 \rightarrow i}^{\text{G}} = \frac{1}{\theta_{\max}}  \int \limits_{A_1}^{A_2} \left( \frac{2}{1+e^{-c_u (\theta (s_0-s_i)-(p_0-p_i))}}-1 \right) d \theta = \\
\left[\frac{A_2-A_1}{\theta_{\max}} + \frac{1}{\theta_{\max}}  \frac{2}{s_0-s_1}\log\frac{e^{-c_1 A_2+c_2} + 1}{e^{-c_1 A_1+c_2} + 1}\right].
\end{array}
\end{equation}

The above expressions cover all four transitions (within the groups $1,2$) for both considered game types. Finally, an inactive customer decides to connect if its utility is above zero (i.e., $\theta > \frac{p_0}{s_0} = \theta_{0, \varnothing}$):
\begin{equation}
\begin{array}{l}
\!\!\!Q_{ \varnothing \rightarrow 0}\! = \!
\left[\frac{\theta_{2, \varnothing}-\theta_{0, \varnothing}}{\theta_{\max}} + \frac{1}{\theta_{\max}}  \frac{2}{s_0-s_1}\log\frac{e^{-c_1 \theta_{2, \varnothing}+c_2} + 1}{e^{-c_1 \theta_{0, \varnothing}+c_2} + 1}\right],
\end{array}
\end{equation}
if $\theta_{0, \varnothing}<\theta_{2, \varnothing}$ or zero, otherwise. Let us also calculate $Q_{0\rightarrow \varnothing } $, since this variable and the one above are the only two parameters of interest for the "inactive" market share:
\begin{equation}
\begin{array}{l}
\!\!\! Q_{  0  \rightarrow \varnothing} \!= \!
\left[\frac{\theta_{0, \varnothing}}{\theta_{\max}} - \frac{1}{\theta_{\max}}  \frac{2}{s_0-s_1}\log\frac{e^{c_1 \theta_{0, \varnothing}-c_2} + 1}{e^{-c_2} + 1}\right].
\end{array}
\end{equation}

\subsubsection{Changing provider due to customer dissatisfaction}
Dissatisfaction of a customer with the initial subscription to the LSP $i$ depends directly on $U(\theta;s_i(T_i^{m}), p_i) - U(\theta;s_j, p_j))$, where $j=0,i$ is the index of the current service provider:
\begin{equation}
 \begin{array}{c}
\! \! \! \! \! \! U(\theta;s_i^{m}, p_i) - U(\theta;s_j, p_j)) >0,\\
   \text{ if }  \theta > \theta_{i,j} ,s_i^{m}>s_j,\text{ or } \theta < \theta_{i,j},s_i^{m}<s_j, j=0,i,
\end{array}
\end{equation}
where $s_i^{m}=s_i(T_i^{m})$ and $\theta_{i,j}  = \frac{p_i-p_j}{s_i^{m}-s_j}$. Here, one may employ all the same derivations as above (even a simpler procedure would suffice due to the fact that $\theta_{i,i} =0$), based on the coefficients from Table~\ref{tab:paramsA}. The four coefficients in question may be obtained according to:
\begin{equation}
\begin{array}{l}
Q_{i\rightarrow j}^{\text{D}} =
\left[\frac{A_2-A_1}{\theta_{\max}} + \frac{1}{\theta_{\max}}  \frac{2}{s_i^m-s_j}\log\frac{e^{-c_1 A_2+c_2} + 1}{e^{-c_1 A_1+c_2} + 1}\right],
\end{array}
\end{equation}
where we do not need to consider the case of negative utility as well as the integral corresponding to the one in the proximate interaction part (see above). We note that at this point all ten coefficients of interest have been obtained. The corresponding system of differential equations is rather cumbersome to solve analytically, but a numerical solution suffices for the purposes of our analysis.

\subsection{Cooperation between the LSPs ("proxy" functionality)}
We finalize this section by modeling possible cooperation between the LSPs. In particular, we assume that based on a certain long-term agreement, mmWave AAPs (drone cells) of the LSP~1 may assist those of the LSP~2 and the other way around, which has been termed previously the \textit{"proxy" AAP} functionality. As a result, the customers of the LSP $i$ may be served by the closest AAP of either LSP. Since "proxy" handover has to be made transparent for the customer device, we assume the use of the \textit{same frequency band} for such operation, but the system \textit{airtime} is shared as per the actual total number of served customers (see Fig.~\ref{fig:beamforming} for details). 

We remind that the shares of the customers are denoted as $x_i,y_i$, as previously. However, these variables correspond to the numbers of devices that may be served by the AAPs of their own LSP. The shares of those that are served by the assisting LSP are denoted as $z^{x}_i,z_i$. The average share $\epsilon$ of "foreign devices" is calculated according to the geometric probability to place a point into the area when assistance is required. Furthermore, the two average spectral efficiencies $\overline \eta_i, i=1,2$ are recalculated according to the new integrated AAP deployment. Given these two spectral efficiencies, we may then calculate the throughput of a customer, if the LSP $i$ is assisted by the LSP $j$, which clearly differs from what we had before. Namely, $T^{z}_i = \frac{B_{j} \overline{\eta}_{i}}{x_i},(i,j)=(1,2),(2,1)$. The operation of the USP does not change when the LSPs cooperate.

Based on the methods similar to those utilized previously, we rewrite the system (\ref{eqn:main_dynamics}) by adding two more equations that reflect the evolution of $z_i$:
\begin{equation}
\left\{ \begin{array}{l}
 \frac{dz_1}{dt} =      (D_1-y_1-z_1) \epsilon \xi  \gamma (y_0+y_1+y_2) \cdot  \tilde Q_{1\rightarrow 0}^{\text{G}}             \\
+ (D_1-y_1-z_1)\epsilon\xi \alpha_c  (1-\gamma)                       									 \\
+   (D_1-y_1-z_1)\epsilon\xi \delta (1-\gamma) \cdot \tilde Q^{\text{D}}_{1\rightarrow0}  							\\
 -   z_1  \xi        \gamma (D_1-y_1-z_1)\cdot  \tilde Q_{0\rightarrow1}^{\text{G}} 						  \\
 - z_1\xi            \alpha_c  (1-\gamma)  											\\
 - z_1 \xi            \delta (1-\gamma)\cdot  \tilde Q^{\text{D}}_{0\rightarrow1},  							  \\
 \\
\frac{dz_2}{dt} =   (D_2-y_2-z_2)\epsilon \xi   \gamma (y_0+y_1+y_2)  \tilde Q_{2\rightarrow0}^{\text{G}} 				 \\
 + (D_2-y_2-z_2)\epsilon \xi              \alpha_c  (1-\gamma) 										 \\
 +  (D_2-y_2-z_2)\epsilon \xi   \delta (1-\gamma) 	\cdot  \tilde Q^{\text{D}}_{2\rightarrow0} 	                       \\
  - z_2  \xi            \gamma (D_2-y_2-z_2) \cdot  \tilde Q_{0\rightarrow2}^{\text{G}} 								\\
 -z_2 \xi          \alpha_c     (1-\gamma)											\\
 - z_2 \xi             \delta (1-\gamma) \cdot  \tilde Q^{\text{D}}_{0\rightarrow2} , 						  \\
\end{array} \right.
\nonumber
\end{equation}

\begin{equation}
\left\{ \begin{array}{l}
\frac{dy_1}{dt} =      (D_1-y_1-z_1)(1-\epsilon) \xi  \gamma (y_0+y_1+y_2) \cdot   Q_{1\rightarrow 0}^{\text{G}}            \\
+ (D_1-y_1-z_1)(1-\epsilon)\xi \alpha_c  (1-\gamma)                       									 \\
+   (D_1-y_1-z_1)(1-\epsilon)\xi \delta (1-\gamma) \cdot  Q^{\text{D}}_{1\rightarrow0}  							\\
 -   y_1  \xi        \gamma (D_1-y_1-z_1)\cdot  Q_{0\rightarrow1}^{\text{G}} 						 \\
 - y_1\xi            \alpha_c  (1-\gamma)  										\\
 - y_1 \xi            \delta (1-\gamma)\cdot  Q^{\text{D}}_{0\rightarrow1},  						 \\
 \\
 \frac{dz_1}{dt} =      (D_1-y_1-z_1) \epsilon \xi  \gamma (y_0+y_1+y_2) \cdot  \tilde Q_{1\rightarrow 0}^{\text{G}}             \\
+ (D_1-y_1-z_1)\epsilon\xi \alpha_c  (1-\gamma)                       									 \\
+   (D_1-y_1-z_1)\epsilon\xi \delta (1-\gamma) \cdot \tilde Q^{\text{D}}_{1\rightarrow0}  							\\
 -   z_1  \xi        \gamma (D_1-y_1-z_1)\cdot  \tilde Q_{0\rightarrow1}^{\text{G}} 						  \\
 - z_1\xi            \alpha_c  (1-\gamma)  											\\
 - z_1 \xi            \delta (1-\gamma)\cdot  \tilde Q^{\text{D}}_{0\rightarrow1},  							  \\
 \\
\frac{dy_2}{dt} =   (D_2-y_2-z_2)(1-\epsilon)\xi   \gamma (y_0+y_1+y_2)   Q_{2\rightarrow0}^{\text{G}} 				\\
 + (D_2-y_2-z_2)(1-\epsilon) \xi              \alpha_c  (1-\gamma) 										 \\
 +  (D_2-y_2-z_2)(1-\epsilon) \xi   \delta (1-\gamma) 	\cdot   Q^{\text{D}}_{2\rightarrow0} 						\\
  - y_2  \xi            \gamma (D_2-y_2-z_2) \cdot   Q_{0\rightarrow2}^{\text{G}} 							\\
 -y_2 \xi          \alpha_c     (1-\gamma)											\\
 - y_2 \xi             \delta (1-\gamma) \cdot  Q^{\text{D}}_{0\rightarrow2} , 						 \\
\\
\frac{dz_2}{dt} =   (D_2-y_2-z_2)\epsilon \xi   \gamma (y_0+y_1+y_2)  \tilde Q_{2\rightarrow0}^{\text{G}} 				 \\
 + (D_2-y_2-z_2)\epsilon \xi              \alpha_c  (1-\gamma) 										 \\
 +  (D_2-y_2-z_2)\epsilon \xi   \delta (1-\gamma) 	\cdot  \tilde Q^{\text{D}}_{2\rightarrow0} 	                       \\
  - z_2  \xi            \gamma (D_2-y_2-z_2) \cdot  \tilde Q_{0\rightarrow2}^{\text{G}} 								\\
 -z_2 \xi          \alpha_c     (1-\gamma)											\\
 - z_2 \xi             \delta (1-\gamma) \cdot  \tilde Q^{\text{D}}_{0\rightarrow2} , 						  \\
\\
\frac{dy_0}{dt} =   (D_0-y_0)\xi\gamma (y_0+y_1+y_2+z_1+z_2) \cdot  Q_{\varnothing \rightarrow 0} 			\\
+(D_0-y_0) \xi \alpha_c  (1-\gamma) 											 \\
-  y_0   \xi \cdot Q_{0\rightarrow \varnothing } ,
\\

\frac{dp_0}{dt}= c_{\text{price}}\frac{d(y_0+y_1+y_2)}{dt}p.\\
\end{array} \right.
\label{eqn:main_dynamics2}
\end{equation}
We note that all of the coefficients without tilde alter according to the updated spectral efficiencies, while those with tilde are calculated by the analogy with the above, but based on different throughputs.


\begin{figure*}[ht!]
\centering
\includegraphics[width=0.8\textwidth]{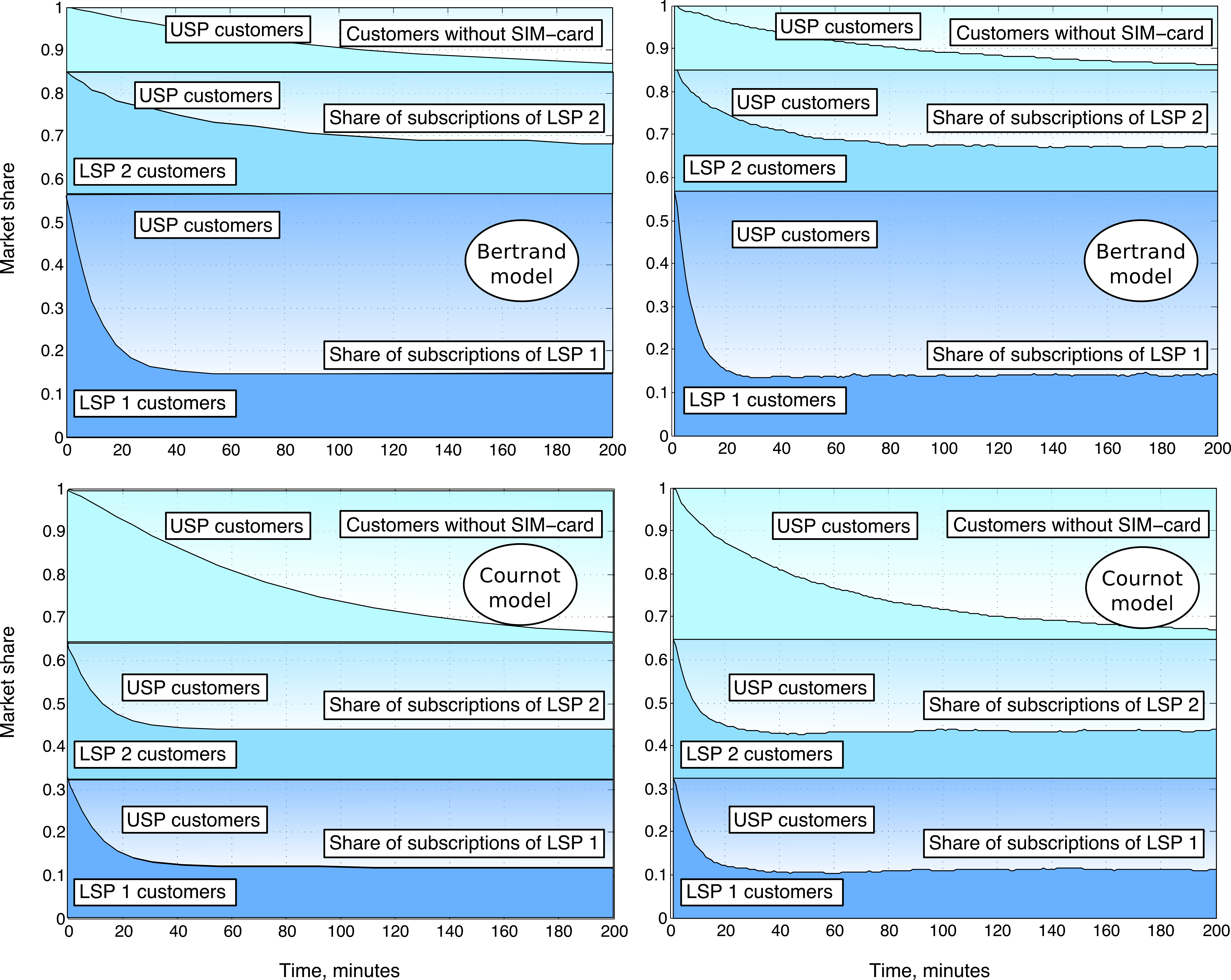}
\protect\caption{Shares of our market in dynamics: the Bertrand (top) and Cournot (bottom) games; analysis (left) and simulation (right).}
\label{fig:results_market}
\end{figure*}

\begin{table}[tbh]
\centering 
\caption{System evaluation parameters}
\begin{tabular}{|c|c|c|c|c|}
\hline \label{tab:params}
Name & Description/definition  & \multicolumn{3}{c|}{Value} \tabularnewline
\hline 
$R$  & Area of interest&  \multicolumn{3}{c|}{200m}\tabularnewline
\hline 
$\mu_0$  & Density of devices&  \multicolumn{3}{c|}{0.2 per m$^2$}\tabularnewline
$\mu$  & Density of users &  \multicolumn{3}{c|}{2 per m$^2$}\tabularnewline
$h_b$  & Average body height &  \multicolumn{3}{c|}{1.7m}\tabularnewline
$r_b$  & Average body radius &  \multicolumn{3}{c|}{0.2m}\tabularnewline
$h_d$  & Average device elevation &  \multicolumn{3}{c|}{1.2m}\tabularnewline
\hline
$\phi$ & Half of aperture&  \multicolumn{3}{c|}{10$^o$}\tabularnewline
$\beta_{\max}$ & Maximum beam inclination&  \multicolumn{3}{c|}{60$^o$}\tabularnewline
$\alpha$  & Antenna "viewing angle" &  \multicolumn{3}{c|}{270$^o$}\tabularnewline
\hline
$h_i$ & AAP altitude & 15m  & 15m & 30m \tabularnewline
$B_i $ & Bandwidth\footnotemark[1] & 2GHz  & 2GHz &  6GHz \\
$p_{tx} $ & Transmit power & \multicolumn{3}{c|}{$200$ mW}\tabularnewline
$f_i $& Frequency &  28GHz & 28GHz & 60GHz \tabularnewline
$N_0$  & Noise plus interference &  \multicolumn{3}{c|}{-80dBm}\tabularnewline
$G_{\text{NLOS}}$  & Reflected beam path decrease &  \multicolumn{3}{c|}{-30dB}\tabularnewline
$\text{SNR}_{\max}$  & Maximum SNR &  \multicolumn{3}{c|}{20dB}\tabularnewline
\hline 
$\theta_{\max}$& Maximum "taste"\footnotemark[2] & \multicolumn{3}{c|}{234.03 }\tabularnewline
$T^{m}$& Maximum throughput & \multicolumn{3}{c|}{100 }\tabularnewline
$a,b,c$& Quality function parameters& \multicolumn{3}{c|}{(12.3600, 0.5619, 0.0098)}\tabularnewline
$\nu$&Costs per unit of demand \footnotemark[3]& \multicolumn{3}{c|}{ 0.0647}\tabularnewline
\hline
 $\gamma$&"Gossiping" factor & \multicolumn{3}{c|}{ 0.05}\tabularnewline
 $\alpha_c$&"Curiosity" factor & \multicolumn{3}{c|}{ 0.05}\tabularnewline
 $\delta$&"Dissatisfaction" factor & \multicolumn{3}{c|}{ 0.1}\tabularnewline
  $c_{\text{price}}$& USP price update& \multicolumn{3}{c|}{ 0.1}\tabularnewline
   $c_{u}$& Utility "impact" coefficient& \multicolumn{3}{c|}{ 0.1}\tabularnewline
\hline
\end{tabular}
  \begin{tablenotes}
\raggedright
\item[1] Maximum bandwidth purchased by the LSP may be decreased according to the difference in the optimal costs at the initial-stage game. For the final throughput, the value in Hz is multiplied by $0.5$ to obtain the "effective bandwidth".\\
\item[2] Derived under the assumption that the "richest" customer is ready to pay not more than $200$\$ per month (i.e., $0.46$ cents per minute).\\
\item[3] Derived assuming that the annual license costs 50 million dollars per 2 GHz.

\end{tablenotes}

\end{table}

\section{Numerical analysis of our system}

\subsection{Implemented modeling environment}
In order to carefully model the dynamic evolution of the system in question, we are not only relying on the above game theoretic tools, but also conduct a comprehensive system-level simulation assessment by building our own evaluation platform with suitable scaling performance to address our large-scale scenario. The developed simulation platform is based on our recent rigorous mmWave study in~\cite{Gal16} and additionally implements several unique features that allow it to capture the intricate interactions between our studied entities:

\begin{itemize}
\item The simulator utilizes a rescaled framing for medium access control, which enables operation across larger time intervals without affecting the system efficiency.
\item The antenna beamforming and body reflection effects are modeled explicitly by applying a mmWave-specific channel model based on the ray-tracing results and real measurements.
\end{itemize}


The investigated short-term dynamics of the aerial mmWave access market is assessed for a characteristic scenario with two LSPs and one USP collocated in a square area of interest with the side length \textcolor{black}{$R = 200$m}. Each of these market players has a fleet of mmWave AAPs (drone cells) uniformly deployed over the area, thus forming a regular hexagonal grid. All AAPs of the LSPs are placed at the same altitude $h_i = 15$~m, while the AAPs of the USP occupy $h_0 = 30$~m, which results in expected numbers of $18$ and $5$ AAPs, respectively. The number of (potentially) active customers is $N=8,000$ (the total number of participants is $80,000$), who are distributed uniformly across this square area. At every instant of time, each active customer is served by its closest AAP of the service provider it is subscribed to. One unit of time corresponds to one minute. For the complete list of the system settings, we refer to Table~\ref{tab:params}.

\subsection{Quantifying market dynamics} \label{sec:5.1}
We begin with investigating the overall dynamics of our subject market, where the initial market shares of the LSPs are determined according to the Bertrand (BM) and the Cournot models (CM) above. These alternatives represent two extreme cases of possible market alignment corresponding to the differentiated vs. identical product offers, respectively, and lead to dissimilar conclusions on system performance. Our first result addresses the case when two independent LSPs observe the consequences of the USP intrusion into "their" equilibrium market and capture the temporal evolution of the \textit{market shares} illustrated here for the scenario with $\alpha_{c}=0.0083$ (equivalent to a period of 2 hours), $\gamma=0.05$, $\delta=0.1$, and the initial price $p_0=p_2/2$. To this end, Fig.~\ref{fig:results_market} indicates time (in minutes) along the horizontal axis, while the vertical axis reports the distribution of the shares of the customers for the market players. On the right side of this plot, we append our simulation results to validate the analysis on the left side.

The key difference that is revealed when comparing the BM and the CM cases of the initial-stage modeling is in their respective market shares (that is, at time $t=0$). For the CM, the LSP~1 and the LSP~2 have equal shares of subscribers (i.e., $0.33$), while the remaining market share comprises the customers with no SIM-cards (not connected to either LSP). In contrast, for the BM, the LSP~1 has a larger initial market share (i.e., $0.58$) and the LSP~2 has a smaller one (i.e., $0.29$). Further, we observe that the evolution of these shares has similar trends for both cases: the USP is able to acquire significant shares of the market not only by activating the customers with no SIM-card, but also winning the customers of the LSPs. The dynamics of this process changes over time by converging to a stable balance for the LSPs in about 1 hour. Given the excellent match between our analytical and simulation results, we focus in the remainder of this section only on the findings of our mathematical analysis.


We continue by investigating the \textit{profit evolution} (including spectrum costs and fees paid by the customers). To this end,  Fig.~\ref{fig:results_param} studies the impact of such factors as consumer "gossiping", "curiosity", and "dissatisfaction" over a realistic range of values (lighter curves). We also report more "extreme" results (darker curves) to highlight the potential amplitude of the profit. Our first observation is that only user "curiosity" has a significant effect on the lowest-profitable USP and the second-profitable LSP~2 (in the BM). Specifically, it increases the profit of an LSP with the growing $\alpha_c$ as well as alters the convergence time. The latter trend holds for the "dissatisfaction" factor as well, but an increase in $\delta$ for the BM (but not the CM) results in a lower profit of the highest-price LSP~1 (since both the announced quality level and the price are higher). This reveals an important difference between the BM and the CM in terms of the USP "intrusion". The last remaining "gossiping" factor $\gamma$ may significantly lower the profit of both LSPs (as shown by the "extreme" line), which implies that the more the customers interact, the higher utility they can achieve. This feature -- in a way similar to implicit customer cooperation -- is clearly beneficial for them, and hence suggests a hypothesis on that some level of cooperation between the LSPs may also bring value. Therefore, we shift our focus to the cooperation of the LSPs in what follows.

\begin{figure*}[ht!]
\centering
\includegraphics[width=2\columnwidth]{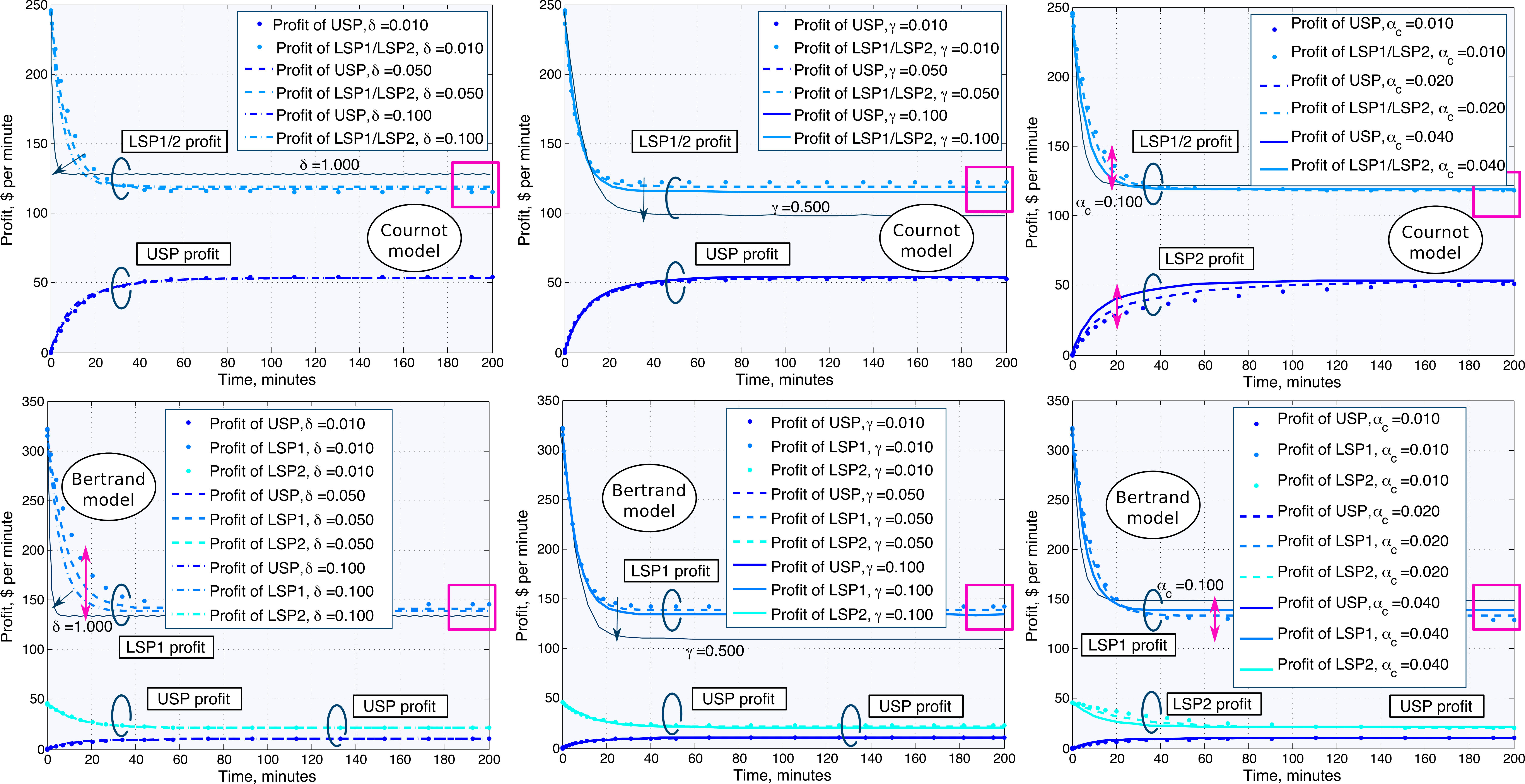}
\protect\caption{Profit of our main market players in dynamics: the Cournot (top) and Bertrand (bottom) games; with "dissatisfaction" (left), "gossiping" (middle), and "curiosity" (right) factors.}
\label{fig:results_param}
\end{figure*}

\subsection{Effects of the LSP cooperation}
\label{sec:5.2}

\begin{figure*}[ht!]
\centering
\includegraphics[width=2.15\columnwidth]{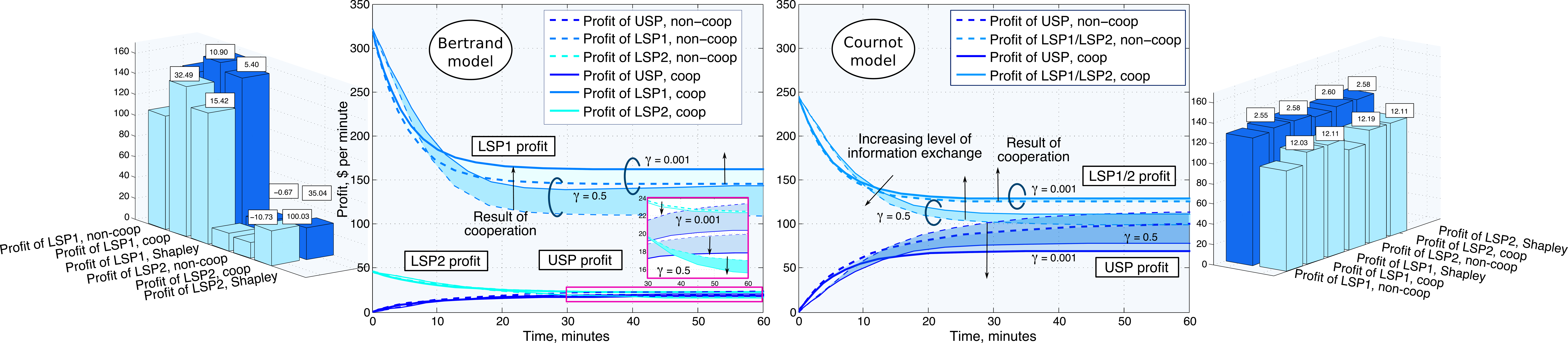}
\protect\caption{Consequences of the LSP cooperation in dynamics: the Bertrand (left) and Cournot (right) games; arrows follow changes in profit due to cooperation.}
\label{fig:results_coop}
\end{figure*}

In this subsection, we characterize the market impact of cooperation between the LSPs, which aims to help serving their customers more efficiently. Accordingly, Fig.~\ref{fig:results_coop} illustrates the temporal evolution of the LSP profits for the two extreme values ($0.001$ and $0.5$) of the "gossiping" factor $\gamma$ with and without cooperation. In the plot, an increase in the profit due to the LSP cooperation is highlighted by the filled areas. Interestingly, we observe differences between the short- and the long-term dynamics for both the BM and the CM cases (see e.g, $t<8$ and $t>8$). It becomes apparent that cooperation always benefits the LSP~1, while the LSP~2 wins only for the CM. On the contrary, for the BM, cooperation may result in a lower LSP~2 profit, which suggests the importance of financial compensation from the winning LSP~1. From the theoretical viewpoint, it translates into the need for a fair utility sharing. Hence, the aggregate profit could be split between the LSPs according to e.g., the Shapley value to enforce fairness of cooperation as it is shown by the bar graphs in Fig.~\ref{fig:results_coop}. 


\section{Summary and conclusions}
The unique modeling framework proposed in this paper makes a decisive step towards understanding the novel market around the mmWave access systems as part of 5G landscape. The unexpected and temporary events featuring masses of people constitute particularly challenging study cases in this area due to unprecedented bandwidth requirements that can only be satisfied with emerging radio technologies, such as aerial mmWave access points. In our systematic performance characterization, a vertically differentiated market proved to be particularly suitable to address this novel setting, by modeling customers that have different preferences with respect to the access service quality. 

To comprehensively model the competition among dissimilar players in the subject market, we utilized the Bertrand and Cournot games, which lead to drastically different performance results. More specifically, in the Cournot model we observe that the market is equally shared between the two LSPs (licensed-band service providers), with subsequent equal profit sharing. Here, the total market profit turns out to be higher than that in the Bertrand model, where instead a clear differentiation between the two LSPs occurs in terms of profits. The other side of the coin is that the Cournot model leads to a lower surplus as well as to a smaller number of served customers, which may become negative factors with respect to the long-term customer loyalty.

Another lesson learned as a result of our analysis is in that the factors of the customer behavior, such as "gossiping", "curiosity", and "dissatisfaction", yield various consequences for the access market dynamics and the resulting profits, as well as depend on whether the Bertrand or Cournot game is played initially. Whenever the customers enjoy high levels of interaction with each other, they can collectively improve their utility by dynamically adapting their service provider choices. At that time, we study the presence of the USP (unlicensed-band service provider) on the market in question, mindful of customer adaptation dynamics. Competing against the USP, the LSPs may also engage into mutual cooperation to better satisfy the customer demands and thus increase their profits. Here, the need for a fair sharing of the resulting profit is essential to guarantee benefits for both cooperating LSPs, and we adopt the Shapley value for that purpose.

Concluding this work, we are convinced that our first-hand game theoretic modeling conducted in this paper should become a useful reference point for future discussions on the dynamic 5G market environments. Multiple new research directions may stem from our present contribution along the lines of dynamic deployment and operation optimization of the AAPs (aerial access points), profit maximization studies across a range of service and business models that enable customer-driven decision making, as well as further game theoretic analysis of alternative taste parameter distributions and market player strategies.

\appendix

\subsection{Blocking probability illustration}

\begin{figure} [!ht]
  \begin{center}
    \includegraphics[width=0.5\textwidth]{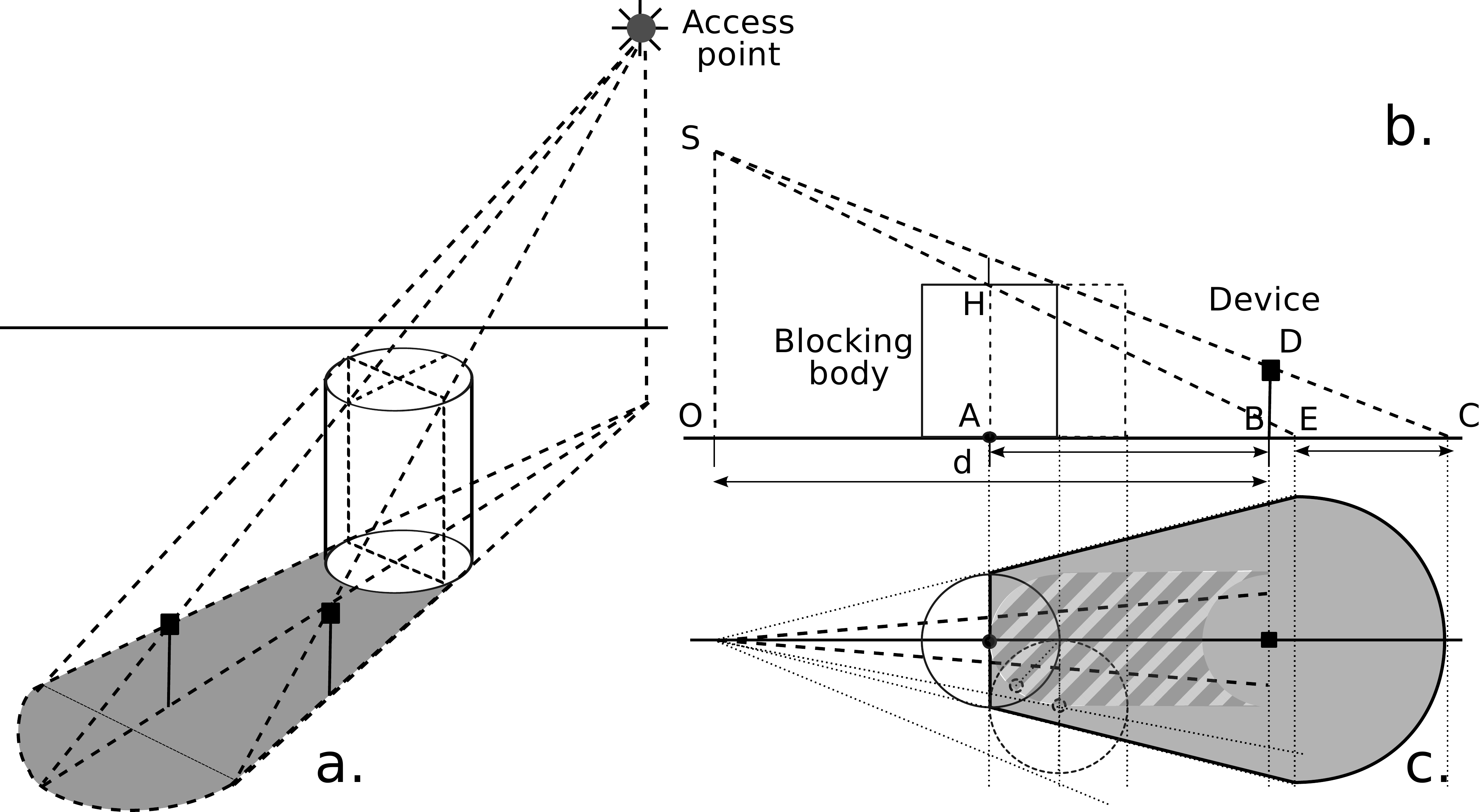}
  \end{center}
  \caption{Illustration of mmWave blockage probability.}
  \label{fig:block}
\end{figure} 

Let us establish here the probability of blockage by a cylindric object, if the device is located at the distance $d$ from the AAP. We assume that the mmWave radio "shadow" is equivalent in its shape and size to the one that visible light creates, and that the LOS is blocked \textit{iff} the axis of symmetry of the radio beam is blocked (see SC in Fig.~\ref{fig:block}.b). Then, the closest blocking body may be located at the point $A$, while the device is placed at the point $B$, and the locus of centers of all the possible blockers is indicated by a hatched region in Fig.~\ref{fig:block}.c of area $S_b \approx 2  r_b \cdot AB $ (we assume that no center can be placed closer than $r_b$). Here, $2  r_b$ is the width of the cylinder representing the human body, $ r_{\max} = d \tan \phi$, $\phi$ is a half of the beam angle, and $AB$ is the longest distance between the device and its blocker. The latter may be calculated based on the properties of similar triangles:
\begin{equation}
 \begin{array}{c}
\frac{h-h_d}{d} = \frac{h_b-h_d}{AB-r_b} ,
\end{array} 
\nonumber
\end{equation} 
where $h$, $h_b$, and $h_d$ are the altitude of AAP, the blocker height, and the device elevation, respectively. Given the above, we obtain:
\begin{equation}
 \begin{array}{c}
 AB = d \frac{h_b-h_d}{h-h_d}+r_b \Rightarrow S_b = 2r_b\cdot d \left( \frac{h_b-h_d}{h-h_d} \right).
\end{array} 
\nonumber
\end{equation} 

We note that the latter expression is valid if the center of the blocker is located between the projection $O$ and the blocked device, and we disregard the cases when the shadow is "egg-shaped". Due to our assumption on the PPP of the blockers, the probability of not being blocked may be produced as follows:
\begin{equation}
 \begin{array}{c}
\Pr\{\text{no blockers in the area}\} \approx e^{-\mu \cdot 2r_b\left( d\frac{h_b-h_d}{h-h_d}\right)},
\end{array} 
\nonumber
\end{equation} 
where $\mu$ is the density of all the participants on the ground of the large-scale event in question.

\subsection{Area of beam on the ground}

\begin{thm}
Consider an AAP at the altitude $h$, its closest device, and a device located at the distance $d$. Then, for the beam aperture $2\phi$, the area covered by the beam under the AAP equals $S_0 = \pi \left (h \tan {(\phi)} \right)^2$. Further, the ratio between the areas $S_0$ and $S(d)$, produced by the device located at the distance $d$, equals: 
\begin{equation}
 \begin{array}{c}
\frac{S(d)}{S_0} = \frac{\cos (\phi)}{\sqrt{\cos^2(\phi)-\frac{d^2}{d^2+{h^2}}}}.
\end{array} 
\end{equation}
\end{thm}
\begin{proof}
Assuming that the area directly under the AAP is a circle or radius $a$, so that $S_0 = \pi a^2$, the area $S$ may be established as $\pi a \cdot b = \pi a^2/ \sqrt{1-e^2}$, where $a = h \tan {(\phi)}$ and $b$ are the minor and the major radius of the ellipse, respectively, $e$ is its eccentricity, and $\phi$ is a half of the beam aperture. We note that the eccentricity of the ellipse representing the result of conic section by a plane $e = \cos \psi/ \cos \phi$, where $\psi$ is the angle between the plane (in our case, the ground) and the cone symmetry axis, that is:
\begin{equation}
 \begin{array}{c}
\tan {(\psi)} = \frac{h}{d} \Rightarrow \cos^2{(\psi)} = \frac{1}{1+\frac{h^2}{d^2}}.
\end{array} 
\end{equation}

Therefore, we may obtain the following ratio:
\begin{equation}
 \begin{array}{c}
\frac{S}{S_0} = \frac{1}{\sqrt{1-\frac{\frac{1}{1+\frac{h^2}{d^2}}}{\cos^2{(\phi)} }}}=
\frac{\cos (\phi)}{\sqrt{\cos^2(\phi)-\frac{d^2}{d^2+{h^2}}}}.
\end{array} 
\end{equation}
\end{proof}

\subsection{Probability of the reflected path}
We assume that in case the LOS link does not exist, the mmWave connection between the AAP and its associated device may be supported via a signal reflected from e.g., another human body. Excluding the blocked device and its possible blocker, we need to estimate whether any other human body is present within the area covered by the directed beam. Utilizing our assumption on the PPP for the devices, we may establish the probability of NLOS path as follows:
\begin{equation}
 \begin{array}{c}
\!q_{\text{NLOS}}(d)\! = \! \Pr \{\text{NLOS at $d$}\}\! =\!
1\!-\!exp\left({- \mu \textcolor{black}{\frac{\pi \left (h\tan{(\phi)} \right)^2 \cdot \cos (\phi)}{\sqrt{\cos^2(\phi)-\frac{d^2}{d^2+{h^2}}} }}} \right)\!.
\end{array} 
\end{equation}

\section*{Acknowledgment}
This work is supported by Intel Corporation and the Academy of Finland. 

\bibliographystyle{ieeetr}
\bibliography{main21_arxiv}

\end{document}